\documentclass[12pt]{article}
\usepackage[T1]{fontenc}
\usepackage[utf8]{inputenc}

\usepackage{amsmath}
\usepackage{amstext}
\usepackage{amssymb}
\usepackage{amsfonts}
\usepackage{geometry}
\usepackage{graphicx}
\usepackage{indentfirst}
\usepackage{appendix}
\usepackage{setspace}
\usepackage{xcolor}
\usepackage{hyperref}
\usepackage{bm}
\usepackage{pgfplots}
\pgfplotsset{compat=1.18}
\usepackage[round,sort,comma,sectionbib]{natbib}

\usepackage[moderate]{savetrees}
\usepackage{times}

\hypersetup{hidelinks,
			bookmarksopen=true,
			pdfstartview={XYZ null null 1.00},
			pdfpagelayout={SinglePage}
}
\newtheorem{theorem}{Theorem}
\newtheorem{assumption}{Assumption}
\newtheorem{definition}{Definition}
\newtheorem{example}{Example}
\newtheorem{lemma}{Lemma}
\newtheorem{proposition}{Proposition}
\newenvironment{proof}{\par{\noindent\textbf{Proof}.\quad}}{\hfill $\blacksquare$\par}

\newcommand{\PreventProofWidow}{\clubpenalty=10000\widowpenalty=10000\relax}

\ifdefined\DeclareUnicodeCharacter
\DeclareUnicodeCharacter{2011}{-}
\DeclareUnicodeCharacter{2010}{-}
\fi
\hypersetup{hypertexnames=false}
\newcommand{\OAref}[1]{Online Appendix~\ref{OA-#1}}
\begin{document}

\title{Generalized Multidimensional Contests with Asymmetric\\ Players: Equilibrium and Optimal Prize Design\thanks{We are grateful to the Editor, Roberto Serrano, the Advisory Editor in charge, and two anonymous reviewers for their insightful comments and suggestions, which significantly improved the quality of the paper. We thank Qian Jiao and Zenan Wu for their helpful comments and suggestions. We thank Junhao Chen for excellent research assistance. All remaining errors are our own. This research is funded by the National Natural Science Foundation of China (No. 72103190, 72473137, 72192801).}}
\author{Siyuan Fan\thanks{Siyuan Fan: Department of Political Science, UC San Diego, 9500 Gilman Drive, La Jolla, CA, 92093.
\emph{E-mail:} sif004@ucsd.edu.}
\and
Zhonghong Kuang\thanks{Zhonghong Kuang: School of Economics and Institute of Future Economy, Renmin University of China, 59 Zhongguancun Street, Beijing, China, 100872. \emph{E-mail:} kuang@ruc.edu.cn.}
\and
Jingfeng Lu\thanks{Jingfeng Lu, Department of Economics, National University of Singapore, 10 Kent Ridge Crescent, Singapore, 119260. \emph{E-mail:} ecsljf@nus.edu.sg.}
}

\date{October 2026}

{\setstretch{1.1}\maketitle}
\vspace{0.4cm}

\begin{abstract}\baselineskip=\normalbaselineskip
\noindent We study $n$-dimensional Tullock contests between two players with heterogeneous effort costs. We establish sharp bounds on discriminatory power $r$ ensuring that equilibrium is unique and pure across prize rules and cost asymmetries: $r\leq1/\lceil n/2\rceil$ when prizes depend only on the number of victories across all dimensions, and $r\leq1/n$ when they may also depend on player identity. An effort-maximizing rule splits the prize unless one player wins by a sufficient margin in the first case, and gives the weaker player a possibly zero headstart in the second. When both bounds hold, restricting prizes to victory counts alone can be costless despite unequal costs, but the percentage loss in total effort weakly increases with asymmetry.


\end{abstract}

\newpage

\section{Introduction}\label{Sec: Intro}
In multidimensional contests, players compete across several dimensions, while prizes depend on aggregate performance. R\&D competition provides a natural interpretation: two firms may compete across several technologies or performance dimensions, each choosing its entire vector of research efforts. We study a stylized contest in which each dimension produces a binary comparison and the research prize depends on the firms' overall records. The prize allocation rule can link incentives across dimensions: the reward from an additional victory may depend on the outcomes of the other comparisons. Differences in effort costs add a further consideration in designing these rewards.

An organizer may value both firms' effort because even unsuccessful research can contribute to subsequent innovation. If the objective is to maximize total research input, which prize allocation rule should the organizer choose? Answering this question requires understanding how equilibrium depends jointly on four factors: the number of battles, discriminatory power within each battle, the aggregate prize rule, and cost asymmetry.

We model two contestants with heterogeneous linear effort costs competing in $n$ disjoint, simultaneous Tullock battles. Discriminatory power $r$ governs how strongly relative effort affects winning chances. A fixed, fully divisible budget is allocated by a rule that depends on the number of battles each player wins and is weakly increasing in victories; the organizer chooses the rule to maximize expected total effort.\footnote{Total effort maximization is a well-studied objective in contest design. See, e.g., \cite{Drugov2020Tournament}, \cite{Fu2020On}, \cite{Toomas2024Optimal}, \cite{Olszewski2020Performance}, and \cite{Zhang2024optimal}.} Identity independence, or neutrality, requires the same reward for the same number of victories; identity-dependent rules may favor a named player.

As a foundation for prize design, we establish sharp bounds guaranteeing a unique, pure equilibrium for every feasible rule and cost ratio: $r\leq1/\lceil n/2\rceil$ under identity independence and $r\leq1/n$ under identity dependence. Equilibrium efforts are uniform across battles, although contestants may deviate in several battles at once. Above either bound, we construct a feasible rule in the corresponding class under which no pure-strategy equilibrium exists when cost asymmetry is sufficiently large. Within the guaranteed regions, uniqueness lets us compare prize rules without an equilibrium-selection assumption. These guarantees shrink as the number of battles grows.

Our main finding characterizes optimal prize design within these equilibrium domains. Under identity independence, an effort-maximizing rule is majority with a tie margin: the whole prize goes to a player who wins by a sufficient margin; otherwise it is split equally. The equal-split region can reward the weaker contestant after a narrow defeat while preserving a further reward for a decisive victory. Under identity dependence, an optimal rule is majority with a headstart: the weaker player receives a possibly zero headstart in wins, and the whole prize goes to the player with more total victories. Both rules maximize the expected marginal prize subject to their design constraints. Identity dependence permits a single reward jump; identity independence requires reflected jumps of equal size. In both design classes, optimal full-prize thresholds for the stronger contestant weakly increase with cost asymmetry within the respective equilibrium domains.

Comparing these designs reveals the effort cost of neutrality. On the common equilibrium domain $0<r\leq1/n$, the proportional effort loss from identity independence weakly increases with cost asymmetry. The loss is zero under symmetry and, with an odd number of battles, remains zero whenever the higher-to-lower marginal-cost ratio is at most $e^2\approx7.39$. Neutrality becomes costly when it prevents concentrating marginal rewards at the most likely pivotal outcomes.

\medskip
\noindent\textbf{Related literature.}
Our paper belongs to the literature on simultaneous-move multidimensional contests.
\cite{szentes2003beyond,szentes2003three} and \cite{Ewerhart2024A} model each dimension as an all-pay auction and characterize equilibrium, while \cite*{Klumpp2006Primaries} use a Tullock contest success function. We extend the framework of \cite*{Klumpp2006Primaries} in two ways: from primarily symmetric linear-cost players to arbitrary player asymmetry, and from majority rule with odd numbers of battles to general prize rules with arbitrary numbers of battles. These generalizations permit prize design, beyond the scope of \cite*{Klumpp2006Primaries}. Our universal bound differs from their rule-specific majority analysis; we also examine the design cost of identity independence. \cite*{Lu2024performance} study two-dimensional Tullock contests with (a)symmetric players and discriminatory power one. They introduce and optimize a grand prize awarded only for winning both dimensions.

A related strand studies team contests with pairwise battles \cite[]{hafner2017tug,Barbieri2024winnereffort}. \cite*{Fu2015Team} characterize equilibrium under simple majority rule. \cite*{Feng2024OptimalTeam} extend history independence to arbitrary prize allocation rules and optimize biased prize allocations, whereas \cite*{kuang2026wp} focus on unbiased rules. These models assign different individuals to different battles. In our model, each of two contestants controls an entire effort vector, so optimality must be verified against deviations across all battles.

Other related models differ in how battles are linked, timed, or organized. In Colonel Blotto games, each battlefield has its own prize, and budget constraints indirectly link battles \cite[]{borel1921theorie,friedman1958game,roberson2006colonel,kvasov2007contests,roberson2012non,kovenock2021generalization}; in our model, battle outcomes jointly determine prize allocation. \mbox{\cite{Feng2018Split}} characterize optimal prize structures in dynamic three-battle Tullock contests. In contrast to this and related sequential-battle models \cite[]{Harris1987Racing,konrad2009multi,gelder2014custer}, we study simultaneous battles. \cite*{xu2022equilibrium} allow players to participate simultaneously in multiple heterogeneous conflicts, while \cite{fu2025decentralized} study decentralized contest design in a network, with each battle's favoritism determined by its independent organizer.

\section{Model}\label{Sec: Settings}

Two risk-neutral, heterogeneous contestants compete simultaneously in $n$ battles. Let $i\in\{A,B\}$ denote the player and $j\in\mathcal{N}=\{1,2,\cdots,n\}$ the battle. Player $i$ exerts effort $x_{i(j)}\geq0$ at cost $c_ix_{i(j)}$. Without loss of generality, player $A$ is (weakly) stronger: $0<c_A\leq c_B$. 
Given an effort profile $(x_{A(j)},x_{B(j)})$ in battle $j$, player $i$'s winning probability follows the generalized Tullock contest success function:
\[
p_{i(j)}(x_{A(j)},x_{B(j)})=
\begin{cases}
\displaystyle\frac{x_{i(j)}^r}{x_{A(j)}^r+x_{B(j)}^r},&x_{A(j)}+x_{B(j)}>0,\\[4pt]
1/2,&x_{A(j)}=x_{B(j)}=0.
\end{cases}
\]
The parameter $r>0$, known as the \emph{discriminatory power}, measures the precision of the winner-selection mechanism, and the term $x^r$ represents the ``impact'' of effort $x$. Conditional on the players' effort vectors, the outcomes of different battles are mutually independent.

There is no common effort budget across battles, but incentives are linked because the reward depends on aggregate performance.

\paragraph{Prize allocation rules.} The organizer seeks to maximize expected total effort subject to a fully divisible unit prize budget, committing to the prize rule before efforts are chosen. Our baseline uses an \textbf{identity-independent} rule $v$, which awards either player $v(k)$ upon winning $k$ battles. We impose the following standard conditions (see e.g., \cite*{Feng2018Split,Feng2024OptimalTeam}).

\begin{assumption}\label{ass:3}
    (i) Non-negativity. $v(k)\geq 0,\forall k=0,1,\cdots,n$.
    
    (ii) Monotonicity. $v(k+1)\geq v(k),\forall k=0,1,\cdots,n-1$.
    
    (iii) Budget Balance. $v(k)+v(n-k)=1,\forall k=0,1,\cdots,n$.
\end{assumption}

A prize allocation rule is \textbf{feasible} if it satisfies \autoref{ass:3}. The prize $v(k)$ can also be read as the player's winning chance of an indivisible prize.
For example, when $n$ is odd, the conventional majority rule, denoted by $v_\text{MR}$, is defined as $v_\text{MR}(k)\triangleq\mathbf{1}(k>\frac{n}{2})$.

\paragraph{Payoffs.} A pure strategy for player $i$ is an effort vector $\bm{x}_i=(x_{i(1)},x_{i(2)},\cdots,x_{i(n)})\in\mathbb{R}_+^n$; a mixed strategy is a probability distribution over $\mathbb{R}_+^n$. Let $\mathbb{P}_i(k\mid n,\bm{x}_A,\bm{x}_B)$ denote the probability that player $i$ wins exactly $k$ battles. Player $i$'s payoff is
\[
\pi_i(\bm x_A,\bm x_B)=\sum_{k=0}^n\mathbb P_i(k\mid n,\bm x_A,\bm x_B)v(k)-c_i\sum_{j=1}^nx_{i(j)}.
\]
\paragraph{Identity-dependent prize rules.}
An identity-dependent rule $(v^A,v^B)$ awards player $i$ the prize $v^i(k)$ upon winning $k$ battles, with $v^A(k)$ and $v^B(k)$ not necessarily equal. Player $i$'s payoff replaces $v(k)$ by $v^i(k)$ above. Feasibility is defined by:

\begin{assumption}\label{ass:ID}
    (i) Non-negativity. $v^i(k)\geq 0,\forall k=0,1,\cdots,n$ and $i=A,B$.

    (ii) Monotonicity. $v^i(k+1)\geq v^i(k),\forall k=0,1,\cdots,n-1$ and $i=A,B$.

    (iii) Budget Balance. $v^A(k)+v^B(n-k)=1,\forall k=0,1,\cdots,n$.
\end{assumption}

\section{Pure-strategy equilibrium}\label{Sec: Equilibrium Analysis}

For equilibrium analysis, it is without loss of generality to normalize nonconstant prize rules by $v(0)=0$ and $v(n)=1$, provided efforts are rescaled accordingly. The constant rule $v(k)=1/2$ is feasible and uniquely induces zero effort. For a nonconstant rule, let $s=v(n)-v(0)\in(0,1]$ and $\widetilde v(k)=[v(k)-v(0)]/s$. We have $\widetilde v(n)=1$ and $\widetilde v(0)=0$, and $\widetilde v$ remains feasible. Under $\widetilde{\bm x}_i=\bm x_i/s$, winning probabilities are unchanged and payoffs satisfy $\pi_i^v(\bm x_A,\bm x_B)=v(0)+s\pi_i^{\widetilde v}(\widetilde{\bm x}_A,\widetilde{\bm x}_B)$. Consequently, pure and mixed equilibria correspond one to one; original efforts are recovered by multiplying normalized efforts by $s$.

\subsection{Reduction to uniform strategies}\label{subsec:preliminary}

Following \cite{Klumpp2006Primaries}, we define strategies allocating equal effort to every battle. To establish equilibrium, however, we must also rule out deviations that allocate effort unevenly.

\begin{definition}
Player $i\in\{A,B\}$ chooses a \textbf{uniform strategy} if he chooses $x_i\in\mathbb R_+$, possibly at random, and then sets $x_{i(j)}=x_i$ for all $j\in\mathcal N$. An equilibrium in which both players choose uniform strategies is called a \textbf{uniform equilibrium}.
\end{definition}

A mixed uniform strategy randomizes the common effort level; it does not randomize independently across battles. The game between two players is denoted by $\mathcal{G}$. Define $\widetilde{\mathcal G}$ as the restricted game of $\mathcal{G}$ where players can only choose uniform strategies.

\begin{proposition}[Uniform-strategy reduction]\label{lem:Uniform Strategy}
For $0<r\leq1$, (i) every equilibrium of $\widetilde{\mathcal G}$ is an equilibrium of $\mathcal G$; (ii) if $\widetilde{\mathcal G}$ has an equilibrium, every equilibrium of $\mathcal G$ is uniform.
\end{proposition}

Complete averaging replaces every component of an effort vector by its mean and preserves cost. Under a nonconstant rule, it strictly improves payoff from a nonuniform vector against an opponent using positive uniform effort. To see this, repeatedly average a largest and a smallest component. Each step weakly improves payoff and is strict when all other components are positive, including at the step that removes the last zero. The sequence converges to the uniform vector, and continuity preserves the strict gain. Against a zero-effort opponent, complete averaging weakly improves payoff. Applying these comparisons realization by realization to mixed strategies establishes part (i).

For part (ii), adding the opponent's effort cost to each player's payoff preserves best responses and, by budget balance, makes the game constant-sum. Equilibrium strategies are therefore \emph{interchangeable} \cite[]{Klumpp2006Primaries,Ewerhart2017Revenue}. If a uniform and a nonuniform equilibrium coexisted, interchangeability would make the nonuniform strategy a best response to the uniform equilibrium opponent. That opponent uses positive effort with positive probability, so complete averaging would strictly improve the nonuniform strategy---a contradiction. \OAref{app:uniform-proof} gives the details.

\subsection{Restricted-game candidate and full-game existence}\label{subsec:pure-strategy-characterization}

Fix a normalized nonconstant rule $v$. We now characterize the only possible pure equilibrium of the restricted game $\widetilde{\mathcal G}$ for every $r>0$. Let $\bar p_A$ and $\bar p_B$ denote the equilibrium winning probabilities in a simultaneous common-value contest, conditional on pure-equilibrium existence.\footnote{Consider a common-value contest with valuation $V>0$ and let $x_A,x_B$ denote players' efforts. Player $i$'s payoff is $\pi_i=$\allowbreak $\frac{x_i^r}{x_A^r+x_B^r}V-c_ix_i$\allowbreak. At a positive pure equilibrium, the first-order condition for player $i$ is $\frac{rx_A^rx_B^r}{(x_A^r+x_B^r)^2}V=$\allowbreak $c_ix_i$\allowbreak. Combining both players' first-order conditions gives $c_Ax_A=c_Bx_B$, implying $\bar p_A=$\allowbreak $\frac{x_A^r}{x_A^r+x_B^r}=$\allowbreak $\frac{c_B^r}{c_A^r+c_B^r}$\allowbreak.} 
They are determined by
 the primitives $r,c_A,c_B$:
 \[
\bar p_A=\frac{c_B^r}{c_A^r+c_B^r},\qquad
\bar p_B=\frac{c_A^r}{c_A^r+c_B^r}.
\]

Let $\Delta v(k)=v(k+1)-v(k)$ and $\theta(k\mid m,p)=\binom mkp^k(1-p)^{m-k}$. At these winning probabilities, the \emph{prize spread}---the expected prize gain from changing one battle's outcome from a loss to a win---is
\begin{equation}\label{eq:prizespread}
\mathcal V=\sum_{k=0}^{n-1}\theta(k\mid n-1,\bar p_A)\Delta v(k).
\end{equation}
For other winning probabilities, write $\mathcal V(p)=\sum_{k=0}^{n-1}\theta(k\mid n-1,p)\Delta v(k)$.

\begin{proposition}[Restricted-game candidate]\label{claim:equilibrium}
For every $r>0$, any pure-strategy equilibrium of the restricted game $\widetilde{\mathcal G}$ under a normalized nonconstant rule is positive, with battle-winning probabilities $\bar p_i$ and common per-battle efforts
\[
x_A^*=\frac{r\bar p_A\bar p_B\mathcal V}{c_A},\qquad
x_B^*=\frac{r\bar p_A\bar p_B\mathcal V}{c_B}.
\]
\end{proposition}

We now prove \autoref{claim:equilibrium}. Both common efforts must be positive. Against zero effort, any positive common effort can be reduced without changing outcomes. If both players choose zero, each receives expected prize $1/2$, whereas a sufficiently small positive common effort wins prize $1$ at arbitrarily small cost.

At positive common efforts, let $p=x_A^r/(x_A^r+x_B^r)$. Conditional independence makes the other-battle win count binomial, and budget balance gives both players the same marginal prize $\mathcal V(p)$. Changing $x_i$ changes effort in all $n$ battles, so the restricted-game first-order condition is
\[
\frac{\partial\pi_i}{\partial x_i}
=\frac{nrp(1-p)\mathcal V(p)}{x_i}-nc_i=0,
\qquad i=A,B.
\]
Thus $c_Ax_A=rp(1-p)\mathcal V(p)=c_Bx_B$.
The effort ratio is therefore $c_B/c_A$, yielding $p=\bar p_A$ and $\mathcal V(p)=\mathcal V$. Substitution gives the stated efforts. This ends the proof of \autoref{claim:equilibrium}.

Budget balance makes the players respond to the same marginal prize. They incur equal expenditure in each battle, so the more efficient player exerts more effort. The rule determines the scale of efforts; the cost ratio determines relative effort and, together with $r$, winning probabilities.

Define the bound\footnote{We thank an anonymous referee for suggesting the improved bound for even $n$.}
\[
b_n=\frac{1}{\lceil n/2\rceil}
=\begin{cases}2/(n+1),&n\text{ odd},\\2/n,&n\text{ even}.\end{cases}
\]

\begin{theorem}[Existence]\label{Theo:existence}
If $0<r\leq b_n$, a uniform pure-strategy equilibrium of the full game $\mathcal G$ exists under every feasible identity-independent rule and every cost ratio.
\end{theorem}

First-order conditions alone do not establish equilibrium, because payoff need not be concave: winning another battle always weakly raises the prize, but the marginal return to effort can rise near a prize threshold before falling. For a nonconstant rule, the proof (\OAref{app:existence-proof}) fixes the opponent's positive uniform candidate effort, substitutes $u=\ln x_i$, and, following \cite{zenou2024sign}, shows that the payoff derivative has the sign of $G_i(u)=\ln[rp(1-p)\mathcal V(p)]-u-\ln c_i$. A bound on the elasticity of the prize spread gives $G_i'(u)<0$ whenever $r\leq b_n$, so payoff rises up to the candidate, where $G_i=0$, and falls above it. This establishes equilibrium in $\widetilde{\mathcal G}$; since $b_n\leq1$, \autoref{lem:Uniform Strategy} extends it to $\mathcal G$.

\paragraph{Applicability and limitations of the bound.}
Our sufficient condition is more restrictive than the condition of \cite{Klumpp2006Primaries}, which applies only to symmetric players under the majority rule. The gap reflects the cost of universality.

With symmetric contestants, majority rule, and fixed positive discriminatory power, \cite{Klumpp2006Primaries} explain why pure-strategy equilibrium fails for sufficiently many battles. At a putative pure equilibrium, each player's expected prize is one half and expenditure cannot exceed this amount, because zero effort guarantees a nonnegative payoff. Increasing every battle effort by a small fixed proportion therefore adds little cost, while giving the deviator a per-battle winning probability strictly above one half, independent of the number of battles. As independent battles accumulate, the law of large numbers makes a majority nearly certain. The prize gain then approaches one half and exceeds the additional cost.

With asymmetric contestants, a persistent per-battle disadvantage makes a majority increasingly unlikely. When the expected prize falls below the candidate effort cost, the weaker player prefers to withdraw, as in \autoref{prop:necessary}. Our guarantee covers every feasible rule and cost ratio, so it requires weaker discrimination as battles accumulate. This restricts the scope of our prize-design results, but pure equilibria may still exist outside this region under particular rules. For example, proportional prizes, $v(k)=k/n$, sustain a pure equilibrium for every number of battles and every cost ratio when $0<r\leq1$.

\subsection{Sharpness and uniqueness}\label{subsec:unique}

The weaker player's option to withdraw establishes sharpness: to be an equilibrium, the candidate in \autoref{claim:equilibrium} must give player $B$ at least the zero payoff available from choosing zero effort.

For even $n$, majority rule splits the prize at an equal score: $v_{\mathrm{even}}(k)=\mathbf 1\{k>n/2\}+\tfrac12\mathbf 1\{k=n/2\}$.

\begin{lemma}[Pure-equilibrium uniformity]\label{lem:pure-uniformity}
For every $r>0$, under any normalized nonconstant identity-independent rule, every pure-strategy equilibrium of the full game has strictly positive efforts that are uniform across battles. Consequently, it coincides with the candidate in \autoref{claim:equilibrium}.
\end{lemma}

Hence ruling out the candidate in \autoref{claim:equilibrium} rules out every pure equilibrium of the full game. The proof of \autoref{lem:pure-uniformity} (\OAref{app:pure-uniformity-proof}) does not use \autoref{lem:Uniform Strategy}.

\begin{proposition}[Sharpness]\label{prop:necessary}
If $r>b_n$, majority rule has no pure-strategy equilibrium when $c_B/c_A$ is sufficiently large.
\end{proposition}

With sufficiently severe asymmetry, the weaker player's small expected prize under majority rule cannot cover the candidate's effort cost. Withdrawal cannot support another pure equilibrium either: the opponent could reduce positive effort without changing outcomes. The bound therefore gives a sharp equilibrium guarantee across all feasible identity-independent rules and cost ratios. \OAref{app:sharpness-proof} gives the proof of \autoref{prop:necessary}.

\begin{proposition}[Uniqueness]\label{pro:unique}
If $0<r\leq b_n$, under every feasible identity-independent rule and every cost ratio, the uniform pure-strategy equilibrium is the unique Nash equilibrium of the full game.
\end{proposition}

The constant rule uniquely induces zero effort. For a nonconstant rule, \autoref{Theo:existence} establishes the candidate as an equilibrium of both games. By \autoref{lem:Uniform Strategy}, every equilibrium of the full game is therefore uniform, possibly mixed. Since the first-order conditions admit a unique solution, no other pure-strategy uniform equilibrium exists. We then rule out mixed-strategy uniform equilibria. Suppose, toward a contradiction, that a mixed-strategy equilibrium exists. By the interchangeability of equilibria, the randomizing player's equilibrium strategy is also a best response to the opponent's positive candidate effort. But, as established in the proof of \autoref{Theo:existence} (\OAref{app:existence-proof}), the candidate uniform effort is the unique best response to the opponent's candidate effort, including among zero and positive efforts. The randomizing player is therefore never indifferent among multiple efforts and cannot randomize in equilibrium---a contradiction.

\begingroup
\setlength{\spaceskip}{\fontdimen2\font plus .2em minus .1em}
\paragraph{Identity-dependent rules.}
Budget balance gives both schedules the common range $s=v^A(n)-v^A(0)=v^B(n)-v^B(0)$, so the preceding normalization argument applies with $\widetilde v^i(k)=[v^i(k)-v^i(0)]/s$ when $s>0$. If $s=0$, zero effort is the unique equilibrium.
\par\endgroup

Budget balance still gives the players a common expected marginal prize,
\[
\mathcal V_{\mathrm{ID}}
=\sum_{k=0}^{n-1}\theta(k\mid n-1,\bar p_A)\Delta v^A(k)
=\sum_{k=0}^{n-1}\theta(k\mid n-1,\bar p_B)\Delta v^B(k),
\qquad \Delta v^i(k)=v^i(k+1)-v^i(k).
\]

For a nonconstant rule, the battle-winning probabilities remain $\bar p_A$ and $\bar p_B$. Without symmetric prize increments, the universal equilibrium bound becomes $1/n$.

\begin{proposition}\label{rem:ID-equilibrium}
(i) If $0<r\leq1/n$, a unique Nash equilibrium exists and is pure under every feasible identity-dependent rule and every cost ratio. For a nonconstant rule, efforts are
\[
x_{i(j)}^*=x_i^*=\frac{r\bar p_A\bar p_B\mathcal V_{\mathrm{ID}}}{c_i},
\qquad i=A,B,\quad j\in\mathcal N.
\]
(ii) If $r>1/n$, the rule that awards $B$ the prize if he wins all $n$ battles and awards $A$ the prize otherwise admits no pure-strategy equilibrium when $c_B/c_A$ is sufficiently large.
\end{proposition}

\OAref{app:ID-proof} gives the proof. Under the rule in part (ii), player $B$'s candidate payoff is $\bar p_B^n(1-nr\bar p_A)$, which is negative for $r>1/n$ when $c_B/c_A$ is sufficiently large. As under identity independence, the bound is sharp as a universal guarantee, although particular rules may admit pure equilibria above it.

\section{Optimal prize design}\label{Sec: Prize Design}

The organizer chooses a feasible prize allocation rule to maximize expected total equilibrium effort,
\[
\mathbf{TE}=\mathbb E\!\left[\sum_{j=1}^n(x_{A(j)}+x_{B(j)})\right].
\]
We impose $0<r\leq1/\lceil n/2\rceil$ for identity-independent rules and $0<r\leq1/n$ for identity-dependent rules. Under these bounds, the equilibrium is unique and pure for every feasible rule and every cost ratio in the respective class (\autoref{Sec: Equilibrium Analysis}).

We may restrict attention to normalized prize rules: $v(0)=0,v(n)=1$, or $v^i(0)=0,v^i(n)=1$ in the identity-dependent case. As shown in \autoref{Sec: Equilibrium Analysis}, normalizing a nonconstant rule multiplies equilibrium effort by $s^{-1}\geq1$. Under the stated bounds, every nonconstant rule induces positive effort, so normalization strictly raises effort when $s<1$ and leaves it unchanged when $s=1$. Constant schedules induce zero effort and cannot be optimal. Normalization is therefore without loss for total-effort maximization.

The key design consideration is the expected marginal reward for a victory. A prize increase between $k$ and $k+1$ victories affects incentives in a given battle only when the contestant wins exactly $k$ other battles. Its contribution is therefore the size of the prize increase weighted by the probability of this pivotal event.

\subsection{Identity-independent prize design}\label{subsec:NI}

For $r\leq b_n$, total effort under an identity-independent prize rule $v$ depends on the spread $\mathcal V(v)$:
\begin{equation}\label{Equa:total effort}
\mathbf{TE}(v)=nr\bar p_A\bar p_B(c_A^{-1}+c_B^{-1})\mathcal V(v).
\end{equation}
Hence, the organizer maximizes $\mathcal{V}(v)$ over $v$. The prize rule changes effort levels but not the cost-adjusted effort ratio or battle-winning probabilities.

We introduce feasible rules that extend majority rule.

\begin{definition}
In a \textbf{majority rule with a tie margin}, a player will be allocated the entire prize if he wins at least $T\in\{\lfloor n/2\rfloor+1,\ldots,n\}$ battles; if neither player wins $T$ battles, the prize will be divided equally.
\end{definition}

The tie margin was introduced by \cite{nalebuff1983prizes} in the context of rank-order tournaments, with subsequent extensions by \cite{imhof2014tournaments,imhof2016ex}. Outside the framework of rank-order tournaments, \cite*{gelder2022all} theoretically examine the tie margin in all-pay auctions, while \cite*{kuang2026wp} experimentally investigate its application in team contests featuring multiple pairwise battles. Under Chile's former two-seat electoral system, a list with two candidates won both seats only if its votes exceeded twice those of the runner-up; otherwise, the two leading lists received one seat each. This provides an institutional analogue of a tie-margin allocation: each list receives one half unless one wins more than two thirds of the two leading lists' combined votes.\footnote{Chile, Law No.~18.700, Article~109 bis, \href{https://www.bcn.cl/leychile/navegar?idNorma=30082&idVersion=2013-06-27}{version of 27 June 2013}.}

We extend the tie-margin idea to multidimensional contests between asymmetric players. A tie margin makes a narrow defeat a draw. For $n=5$, a tie margin of two splits the prize at $3$--$2$, where the lead is one, and awards the whole prize at $4$--$1$, where the lead is three. In general, threshold $T$ corresponds to a tie margin of $2T-n-1$ victories; since every lead has the parity of $n$ and the margin does not, no score falls on the boundary. For simplicity, we use the \emph{winning threshold} $T$ throughout; in the example, $T=4$.

Recall that $\theta(k\mid n-1,p)$ is the probability of winning $k$ out of $n-1$ battles when the per-battle winning probability is $p$. Define $g(k,p)=\theta(k\mid n-1,p)+\theta(n-1-k\mid n-1,p)$.

\begin{theorem}\label{thm:tiemargin}
For $0<r\leq1/\lceil n/2\rceil$, the optimal identity-independent rules are exactly the convex combinations of majority rules with a tie margin and thresholds $T=k+1$, where $k$ maximizes $g(k,\bar p_A)$ over $\{\lfloor n/2\rfloor,\ldots,n-1\}$.
\end{theorem}

The optimal rule is unique if the maximizer is unique. The smallest admissible index, $k=\lfloor n/2\rfloor$, gives conventional majority: a player wins the entire prize if he wins more battles than his opponent, and the prize is split equally at a tied score when $n$ is even. With symmetric players and odd $n$, conventional majority is the unique optimal rule, because $g(k,1/2)=2\binom{n-1}{k}2^{1-n}$ is then uniquely maximized over the upper half at the central index $(n-1)/2$.

A nonempty tie region creates two layers of strategic incentives: avoiding defeat to secure a draw, and striving to turn a draw into a decisive win. The associated reward increases are each one half, at reflected performance thresholds.

We now prove \autoref{thm:tiemargin}. Every optimal rule is normalized, as shown at the beginning of this section, and by \autoref{Equa:total effort}, maximizing $\mathbf{TE}$ over normalized rules is equivalent to maximizing $\mathcal{V}(v)$. Let $m=n-1$ and $a_k=v(k+1)-v(k)$. Monotonicity, normalized endpoints, and identity independence give $a_k\geq0$, $\sum_{k=0}^{m}a_k=1$, and $a_k=a_{m-k}$. For $\lceil m/2\rceil\leq k\leq m$, let $\lambda_k=2a_k$ if $k>m/2$, and $\lambda_k=a_k$ at the central singleton. Then $\lambda_k\geq0$, $\sum_k\lambda_k=1$, and
\[
\mathcal V(v)=\sum_{k=\lceil m/2\rceil}^{m}\lambda_k\frac{g(k,\bar p_A)}2.
\]
This expression also holds at the singleton because $g(k,\bar p_A)=2\theta(k\mid m,\bar p_A)$ there. It is maximized exactly when all positive $\lambda_k$ are supported on maximizing indices. The extreme rule associated with index $k$ is
\[
v_k(t)=\tfrac12\mathbf 1\{t\geq n-k\}+\tfrac12\mathbf 1\{t\geq k+1\},
\]
the tie-margin rule with $T=k+1$. Since $v(0)=0$ and the increments of $\sum_k\lambda_kv_k$ equal $a_k$, we have $v=\sum_k\lambda_kv_k$. Hence the optimal rules are exactly the convex combinations of the rules $v_k$ with maximizing indices $k$. This ends the proof of \autoref{thm:tiemargin}.

\paragraph{Implication.}
Aircraft-engine development illustrates multidimensional competition. General Electric and Pratt \& Whitney developed competing engines that the U.S. Air Force evaluated on capability, readiness and support, life-cycle cost, and program adequacy and competition. The Air Force ultimately split the fiscal 1985 engine purchase, 75 percent to General Electric and 25 percent to Pratt \& Whitney.\footnote{U.S. General Accounting Office, \href{https://www.gao.gov/assets/nsiad-84-104.pdf}{\emph{Review of the Alternate Fighter Engine Competition}}, GAO/NSIAD-84-104 (1984), pp.~1--2, 4--5, 8.} In a stylized interpretation of this setting, the same two firms compete in every dimension, each choosing its entire vector of research efforts; each technological comparison gives one victory to the firm with the better solution, and the research prize depends on their aggregate records. For an organizer seeking to maximize total research effort, \autoref{thm:tiemargin} then calls for an equal split unless one firm achieves a sufficiently large lead. The Air Force's split was discretionary and unequal, so this example illustrates rather than implements the rule.

\subsection{Identity-dependent prize design}\label{subsec:ID}
In R\&D procurement, organizers may be able and willing to treat contestants differently.

\paragraph{Optimal rule.}
For $r\leq1/n$, \autoref{rem:ID-equilibrium} gives uniform equilibrium efforts $x_A^*$ and $x_B^*$. Replacing $\mathcal V$ by $\mathcal V_{\mathrm{ID}}$ in \autoref{Equa:total effort} shows that maximizing $\mathbf{TE}_{\mathrm{ID}}(v)$ is equivalent to maximizing $\mathcal{V}_{\mathrm{ID}}(v)$. We introduce the following rules, studied by \cite*{Feng2024OptimalTeam} in a team-contest environment.

\begin{definition}
In a \textbf{majority rule with a headstart}, player $A$ will be allocated the entire prize if he wins at least $\hat T\in\{1,\ldots,n\}$ battles; otherwise, the entire prize is allocated to player $B$. Equivalently, player $B$ is given a headstart in the form of $2\hat T-n-1$ initial wins, and the entire prize is awarded to the player with the higher number of wins. A negative headstart favors $A$.
\end{definition}

A headstart changes only the final allocation. It requires player $A$ alone to lead by more than $2\hat T-n-1$ victories, whereas a tie-margin rule requires either player to lead by more than $2T-n-1$ to claim the whole prize; a headstart therefore places the entire marginal prize at a single threshold.

\begin{theorem}\label{pro:headstart}
For $0<r\leq1/n$, the optimal identity-dependent rules are exactly the convex combinations of headstart rules with thresholds $\hat T=k+1$, where $k$ is a mode of $\theta(k\mid n-1,\bar p_A)$.
\end{theorem}

An optimal threshold is therefore $\hat T=\lfloor n\bar p_A\rfloor+1$. This rule is uniquely optimal unless $n\bar p_A$ is an integer $\ell$. In that case, thresholds $\ell$ and $\ell+1$, and all convex combinations of the two rules, are optimal. Under symmetry with even $n$, these two threshold rules favor different players but generate the same total effort.

To see why, note that every optimal rule is normalized and write such a rule in terms of its nonnegative prize increments $a_k=\Delta v^A(k)$, with $\sum_ka_k=1$. The objective $\mathcal V_{\mathrm{ID}}=\sum_ka_k\theta(k\mid n-1,\bar p_A)$ is a weighted average of binomial probabilities, so it is maximized precisely when positive increments occur only at binomial modes. A unit increment at $k$ gives the headstart rule with threshold $k+1$. Since $v^A(t)=\sum_ka_k\mathbf 1\{t\geq k+1\}$ and budget balance determines $v^B$, this establishes the characterization.

\paragraph{Implication.}
In the same stylized R\&D contest, an optimal headstart credits the less efficient firm with initial victories; when positive, the credit lowers the number of actual victories it needs to obtain the research prize. Firm-specific score credits are used in R\&D funding: New Jersey's Clean Tech R\&D Pilot Seed Grant, Round~3, adds a 10-point bonus for applicants that meet the qualifying score and have not previously received a CSIT grant or voucher.\footnote{New Jersey, \href{https://www.njeda.gov/wp-content/uploads/2023/12/120523_FINAL_Clean-Tech-Seed-Grant-Pilot-Program-NOFA_Round-3.pdf}{Clean Tech R\&D Pilot Seed Grant, Round 3} (2023), pp.~2, 6.}

\subsection{Optimal thresholds and the cost of neutrality}\label{subsec:statics}

We first study how asymmetry, measured by $\bar p_A$, affects optimal prize thresholds.

\paragraph{Comparative statics of optimal thresholds.}

\begin{proposition}\label{Lemma for p}
Fix $n$. For parts (i) and (ii), let $1/2\leq p_1<p_2<1$.

(i) Under identity independence, let $0<r\leq1/\lceil n/2\rceil$. If $T_1$ and $T_2$ are any optimal tie-margin thresholds at $\bar p_A=p_1$ and $\bar p_A=p_2$, respectively, then $T_1\leq T_2$.

(ii) Under identity dependence, let $0<r\leq1/n$. If $\hat T_1$ and $\hat T_2$ are any optimal headstart thresholds at $\bar p_A=p_1$ and $\bar p_A=p_2$, respectively, then $\hat T_1\leq\hat T_2$.

(iii) Let $0<r\leq1/n$. At any common winning probability $\bar p_A=p\in(1/2,1)$, any optimal tie-margin threshold $T$ and headstart threshold $\hat T$ satisfy $T\leq\hat T$. At $p=1/2$, both equal $(n+1)/2$ if $n$ is odd; if $n$ is even, $T=n/2+1$ and $\hat T\in\{n/2,n/2+1\}$.
\end{proposition}

Parts (i) and (ii) imply that, as player $A$ becomes more likely to win each battle, optimal rules in both classes require weakly more victories from $A$ to claim the full prize. A higher $T$ widens the tie region, protecting player $B$ against narrow defeats; a higher $\hat T$ increases $B$'s headstart. Because $\bar p_A=1/[1+(c_A/c_B)^r]$, a higher cost ratio $c_B/c_A$, or a higher discriminatory power $r$ when $c_A<c_B$, raises $\bar p_A$. Both changes therefore weakly increase optimal thresholds, provided $r$ remains within the respective equilibrium region. By part (iii), under asymmetry, an optimal tie-margin rule never requires more victories for the full prize than an optimal headstart rule.

As $A$'s battle-winning probability $p=\bar p_A$ rises, his win count in the other $n-1$ battles shifts upward. A tie-margin threshold $T$ yields expected marginal prize $g(T-1,p)/2$, whereas a headstart threshold $\hat T$ yields $\theta(\hat T-1\mid n-1,p)$. Within each class, the ratio of a higher threshold's incentive value to a lower threshold's strictly increases with $p\in[1/2,1)$, so optimal thresholds cannot fall, even when several are optimal. This single-crossing argument follows \cite{Milgrom1994Monotone}; Appendix~\ref{app:statics-proof} proves all three parts.

\paragraph{The effort cost of neutrality.}
The cost of neutrality is the effort forgone by requiring the same prize for the same number of victories, rather than allowing identity-dependent prize rules.

On the common equilibrium region $r\leq1/n$, let $\mathbf{TE}_{\mathrm{NI}}^*$ and $\mathbf{TE}_{\mathrm{ID}}^*$ denote maximal total effort under identity-independent and identity-dependent rules, respectively, and put $m=n-1$ and $\theta_k=\theta(k\mid m,\bar p_A)$. Then
\begin{equation}\label{eq:relative-retention}
\rho(\bar p_A)\equiv
\frac{\mathbf{TE}_{\mathrm{NI}}^*}{\mathbf{TE}_{\mathrm{ID}}^*}
=\frac{\max_{\lceil m/2\rceil\leq k\leq m}(\theta_k+\theta_{m-k})}
{2\max_{0\leq k\leq m}\theta_k}\in[1/2,1].
\end{equation}
We measure the cost of neutrality by the proportional loss $1-\rho(\bar p_A)$. To understand the bounds, think of the unit prize as two half-prize increments. Identity-dependent design can place both halves at the threshold where an additional victory is most likely to change the reward. A neutral rule can place one half there, but must put the other half at the reflected threshold to treat the contestants alike. Since effort is proportional to the expected marginal prize, neutrality retains at least half the unrestricted effort, and it is costless when both increments can be placed at incentive-maximizing thresholds.

\paragraph{Comparative statics of the cost of neutrality.}
\begin{proposition}[Relative cost of neutrality]\label{prop:relative-cost}

Fix $n\geq2$ and $0<r\leq1/n$. Relative effort under neutrality, $\rho(p)$, is weakly decreasing in $p\in[1/2,1)$, with $\rho(1/2)=1$ and $\lim_{p\uparrow1}\rho(p)=1/2$. Neutrality is costless exactly at $p=1/2$ for even $n$, and exactly on $[1/2,(n+1)/(2n)]$ for odd $n$.

\end{proposition}

For modest asymmetry and odd $n$, the most likely pivotal count is still central: placing all marginal reward there satisfies neutrality at no cost. In terms of primitives, for odd $n\geq3$, $\rho=1$ whenever $c_B/c_A\leq[(n+1)/(n-1)]^{1/r}$. Since $r\leq1/n$, this cutoff exceeds $e^2\approx7.39$, so neutrality is costless whenever the higher-to-lower marginal-cost ratio is at most $e^2$.\footnote{Since $1/r\geq n$, the cutoff is at least $[(n+1)/(n-1)]^n$. The inequality $\ln(1+x)>2x/(2+x)$ for $x>0$, evaluated at $x=2/(n-1)$, implies $n\ln[(n+1)/(n-1)]>2$.} As asymmetry grows, the most valuable reward locations move away from the center, and identity independence forces part of the reward onto a less likely reflected count; with extreme asymmetry, that count is rarely pivotal and the loss approaches one half.

Appendix~\ref{app:relative-cost-proof} proves the proposition. The result concerns proportional loss: the absolute effort difference need not increase with asymmetry, because overall effort can fall as one contestant becomes very costly.

\begin{example}[Seven battles]\label{ex:seven-battles}
Consider a 7-dimensional contest with $0<r\leq1/7$, so under every rule in either design class and every cost ratio the equilibrium is unique and pure. \autoref{fig:design-comparison} compares their optimal thresholds and the resulting proportional effort loss. The horizontal axis is the strong player's equilibrium battle-winning probability, $p=\bar p_A$, which summarizes cost asymmetry and discrimination.

\begin{figure}[tbp]
\centering
\begin{tikzpicture}
\begin{axis}[
    name=thresholds,
    width=.86\linewidth, height=3.15cm, scale only axis,
    xmin=.5, xmax=1, ymin=3.75, ymax=7.65,
    axis lines=left, axis line style={black},
    tick align=outside, tick style={black},
    xtick={.5,.6,.7,.8,.9,1}, xticklabels={},
    ytick={4,5,6,7}, yticklabels={$4$,$5$,$6$,$7$},
    ymajorgrids=true, grid style={black!18,line width=.3pt},
    tick label style={font=\small},
    ylabel={Threshold}, ylabel style={font=\small},
    title={(a) Optimal full-prize thresholds for $A$},
    title style={at={(0,1.03)},anchor=south west,font=\small},
    legend style={at={(0.01,.98)},anchor=north west,draw=none,
        fill=white,font=\footnotesize,legend columns=2,column sep=9pt},
    legend cell align=left,
    clip=false,
]
\addlegendimage{black,line width=1.1pt,mark=*,mark size=2pt}
\addlegendentry{Tie-margin threshold: $T$}
\addlegendimage{black,line width=1.1pt,dash pattern=on 5pt off 3pt,mark=square*,mark size=1.8pt,mark options={solid}}
\addlegendentry{Headstart threshold: $\hat T$}
\pgfplotsset{
    neutral threshold/.style={black,line width=1.1pt,y filter/.expression={y+.10},forget plot},
    dependent threshold/.style={black,line width=1.1pt,dash pattern=on 5pt off 3pt,y filter/.expression={y-.10},forget plot}
}
\addplot[neutral threshold] coordinates {(.5,4) (.688982236505,4)};
\addplot[neutral threshold] coordinates {(.688982236505,5) (.735268420811,5)};
\addplot[neutral threshold] coordinates {(.735268420811,6) (.857234383846,6)};
\addplot[neutral threshold] coordinates {(.857234383846,7) (1,7)};
\addplot[neutral threshold,only marks,mark=*,mark size=2pt] coordinates {
    (.5,4) (.688982236505,4) (.688982236505,5)
    (.735268420811,5) (.735268420811,6)
    (.857234383846,6) (.857234383846,7)};
\addplot[dependent threshold] coordinates {(.5,4) (.571428571429,4)};
\addplot[dependent threshold] coordinates {(.571428571429,5) (.714285714286,5)};
\addplot[dependent threshold] coordinates {(.714285714286,6) (.857142857143,6)};
\addplot[dependent threshold] coordinates {(.857142857143,7) (1,7)};
\addplot[dependent threshold,only marks,mark=square*,mark size=1.8pt,mark options={solid}] coordinates {
    (.5,4) (.571428571429,4) (.571428571429,5)
    (.714285714286,5) (.714285714286,6)
    (.857142857143,6) (.857142857143,7)};
\addplot[neutral threshold,only marks,mark=o,mark size=2pt,mark options={solid,fill=white}] coordinates {(1,7)};
\addplot[dependent threshold,only marks,mark=square,mark size=1.8pt,mark options={solid,fill=white}] coordinates {(1,7)};
\end{axis}
\begin{axis}[
    at={(thresholds.south west)},anchor=north west,yshift=-1.12cm,
    width=.86\linewidth, height=3.15cm, scale only axis,
    xmin=.5, xmax=1, ymin=-3, ymax=54,
    axis lines=left, axis line style={black!65},
    tick align=outside, tick style={black!65},
    xtick={.5,.6,.7,.8,.9,1},
    xticklabels={$0.5$,$0.6$,$0.7$,$0.8$,$0.9$,$1$},
    ytick={0,25,50}, yticklabels={$0$,$25$,$50$},
    tick label style={font=\small},
    xlabel={$A$'s equilibrium battle-winning probability, $\bar p_A$},
    xlabel style={font=\small},
    ylabel={Effort loss (\%)}, ylabel style={font=\small},
    title={(b) Proportional effort loss from identity independence},
    title style={at={(0,1.03)},anchor=south west,font=\small},
    clip=false,
]
\addplot[black!25,line width=.45pt,densely dotted,forget plot] coordinates {(.5,50) (1,50)};
\addplot[black,line width=1.05pt,no marks] table[row sep=\\,x=p,y=loss] {
p loss\\
0.500000000000 0.000000000000\\
0.501000000000 0.000000000000\\
0.502000000000 0.000000000000\\
0.503000000000 0.000000000000\\
0.504000000000 0.000000000000\\
0.505000000000 0.000000000000\\
0.506000000000 0.000000000000\\
0.507000000000 0.000000000000\\
0.508000000000 0.000000000000\\
0.509000000000 0.000000000000\\
0.510000000000 0.000000000000\\
0.511000000000 0.000000000000\\
0.512000000000 0.000000000000\\
0.513000000000 0.000000000000\\
0.514000000000 0.000000000000\\
0.515000000000 0.000000000000\\
0.516000000000 0.000000000000\\
0.517000000000 0.000000000000\\
0.518000000000 0.000000000000\\
0.519000000000 0.000000000000\\
0.520000000000 0.000000000000\\
0.521000000000 0.000000000000\\
0.522000000000 0.000000000000\\
0.523000000000 0.000000000000\\
0.524000000000 0.000000000000\\
0.525000000000 0.000000000000\\
0.526000000000 0.000000000000\\
0.527000000000 0.000000000000\\
0.528000000000 0.000000000000\\
0.529000000000 0.000000000000\\
0.530000000000 0.000000000000\\
0.531000000000 0.000000000000\\
0.532000000000 0.000000000000\\
0.533000000000 0.000000000000\\
0.534000000000 0.000000000000\\
0.535000000000 0.000000000000\\
0.536000000000 0.000000000000\\
0.537000000000 0.000000000000\\
0.538000000000 0.000000000000\\
0.539000000000 0.000000000000\\
0.540000000000 0.000000000000\\
0.541000000000 0.000000000000\\
0.542000000000 0.000000000000\\
0.543000000000 0.000000000000\\
0.544000000000 0.000000000000\\
0.545000000000 0.000000000000\\
0.546000000000 0.000000000000\\
0.547000000000 0.000000000000\\
0.548000000000 0.000000000000\\
0.549000000000 0.000000000000\\
0.550000000000 0.000000000000\\
0.551000000000 0.000000000000\\
0.552000000000 0.000000000000\\
0.553000000000 0.000000000000\\
0.554000000000 0.000000000000\\
0.555000000000 0.000000000000\\
0.556000000000 0.000000000000\\
0.557000000000 0.000000000000\\
0.558000000000 0.000000000000\\
0.559000000000 0.000000000000\\
0.560000000000 0.000000000000\\
0.561000000000 0.000000000000\\
0.562000000000 0.000000000000\\
0.563000000000 0.000000000000\\
0.564000000000 0.000000000000\\
0.565000000000 0.000000000000\\
0.566000000000 0.000000000000\\
0.567000000000 0.000000000000\\
0.568000000000 0.000000000000\\
0.569000000000 0.000000000000\\
0.570000000000 0.000000000000\\
0.571000000000 0.000000000000\\
0.571428571429 0.000000000000\\
0.572000000000 0.233100233100\\
0.573000000000 0.639906922629\\
0.574000000000 1.045296167247\\
0.575000000000 1.449275362319\\
0.576000000000 1.851851851852\\
0.577000000000 2.253032928943\\
0.578000000000 2.652825836217\\
0.579000000000 3.051237766264\\
0.580000000000 3.448275862069\\
0.581000000000 3.843947217441\\
0.582000000000 4.238258877434\\
0.583000000000 4.631217838765\\
0.584000000000 5.022831050228\\
0.585000000000 5.413105413105\\
0.586000000000 5.802047781570\\
0.587000000000 6.189664963089\\
0.588000000000 6.575963718821\\
0.589000000000 6.960950764007\\
0.590000000000 7.344632768362\\
0.591000000000 7.727016356458\\
0.592000000000 8.108108108108\\
0.593000000000 8.487914558741\\
0.594000000000 8.866442199776\\
0.595000000000 9.243697478992\\
0.596000000000 9.619686800895\\
0.597000000000 9.994416527080\\
0.598000000000 10.367892976589\\
0.599000000000 10.740122426266\\
0.600000000000 11.111111111111\\
0.601000000000 11.480865224626\\
0.602000000000 11.849390919158\\
0.603000000000 12.216694306247\\
0.604000000000 12.582781456954\\
0.605000000000 12.947658402204\\
0.606000000000 13.311331133113\\
0.607000000000 13.673805601318\\
0.608000000000 14.035087719298\\
0.609000000000 14.395183360701\\
0.610000000000 14.754098360656\\
0.611000000000 15.111838516094\\
0.612000000000 15.468409586057\\
0.613000000000 15.823817292007\\
0.614000000000 16.178067318132\\
0.615000000000 16.531165311653\\
0.616000000000 16.883116883117\\
0.617000000000 17.233927606699\\
0.618000000000 17.583603020496\\
0.619000000000 17.932148626817\\
0.620000000000 18.279569892473\\
0.621000000000 18.625872249061\\
0.622000000000 18.971061093248\\
0.623000000000 19.315141787052\\
0.624000000000 19.658119658120\\
0.625000000000 20.000000000000\\
0.626000000000 20.340788072417\\
0.627000000000 20.680489101542\\
0.628000000000 21.019108280255\\
0.629000000000 21.356650768415\\
0.630000000000 21.693121693122\\
0.631000000000 22.028526148970\\
0.632000000000 22.362869198312\\
0.633000000000 22.696155871511\\
0.634000000000 23.028391167192\\
0.635000000000 23.359580052493\\
0.636000000000 23.689727463312\\
0.637000000000 24.018838304553\\
0.638000000000 24.346917450366\\
0.639000000000 24.673969744392\\
0.640000000000 25.000000000000\\
0.641000000000 25.325013000520\\
0.642000000000 25.649013499481\\
0.643000000000 25.972006220840\\
0.644000000000 26.293995859213\\
0.645000000000 26.614987080103\\
0.646000000000 26.934984520124\\
0.647000000000 27.253992787223\\
0.648000000000 27.572016460905\\
0.649000000000 27.889060092450\\
0.650000000000 28.205128205128\\
0.651000000000 28.520225294419\\
0.652000000000 28.834355828221\\
0.653000000000 29.147524247065\\
0.654000000000 29.459734964322\\
0.655000000000 29.770992366412\\
0.656000000000 30.081300813008\\
0.657000000000 30.390664637240\\
0.658000000000 30.699088145897\\
0.659000000000 31.006575619626\\
0.660000000000 31.313131313131\\
0.661000000000 31.618759455371\\
0.662000000000 31.923464249748\\
0.663000000000 32.227249874309\\
0.664000000000 32.530120481928\\
0.665000000000 32.832080200501\\
0.666000000000 33.133133133133\\
0.667000000000 33.433283358321\\
0.668000000000 33.732534930140\\
0.669000000000 34.030891878426\\
0.670000000000 34.328358208955\\
0.671000000000 34.624937903626\\
0.672000000000 34.920634920635\\
0.673000000000 35.215453194651\\
0.674000000000 35.509396636993\\
0.675000000000 35.802469135802\\
0.676000000000 36.094674556213\\
0.677000000000 36.386016740522\\
0.678000000000 36.676499508358\\
0.679000000000 36.966126656848\\
0.680000000000 37.254901960784\\
0.681000000000 37.542829172785\\
0.682000000000 37.829912023460\\
0.683000000000 38.116154221571\\
0.684000000000 38.401559454191\\
0.685000000000 38.686131386861\\
0.686000000000 38.969873663751\\
0.687000000000 39.252789907812\\
0.688000000000 39.534883720930\\
0.688982236505 39.811169380648\\
0.689000000000 39.812858500045\\
0.690000000000 39.907582440664\\
0.691000000000 40.001591686371\\
0.692000000000 40.094891242608\\
0.693000000000 40.187486074932\\
0.694000000000 40.279381109386\\
0.695000000000 40.370581232855\\
0.696000000000 40.461091293434\\
0.697000000000 40.550916100772\\
0.698000000000 40.640060426433\\
0.699000000000 40.728529004239\\
0.700000000000 40.816326530612\\
0.701000000000 40.903457664921\\
0.702000000000 40.989927029813\\
0.703000000000 41.075739211548\\
0.704000000000 41.160898760331\\
0.705000000000 41.245410190634\\
0.706000000000 41.329277981526\\
0.707000000000 41.412506576986\\
0.708000000000 41.495100386224\\
0.709000000000 41.577063783990\\
0.710000000000 41.658401110891\\
0.711000000000 41.739116673689\\
0.712000000000 41.819214745613\\
0.713000000000 41.898699566653\\
0.714000000000 41.977575343863\\
0.714285714286 42.000000000000\\
0.715000000000 42.258448187833\\
0.716000000000 42.618431098842\\
0.717000000000 42.976279527958\\
0.718000000000 43.332012630079\\
0.719000000000 43.685649366044\\
0.720000000000 44.037208504801\\
0.721000000000 44.386708625547\\
0.722000000000 44.734168119840\\
0.723000000000 45.079605193689\\
0.724000000000 45.423037869614\\
0.725000000000 45.764483988683\\
0.726000000000 46.103961212529\\
0.727000000000 46.441487025331\\
0.728000000000 46.777078735779\\
0.729000000000 47.110753479019\\
0.730000000000 47.442528218561\\
0.731000000000 47.772419748178\\
0.732000000000 48.100444693772\\
0.733000000000 48.426619515223\\
0.734000000000 48.750960508215\\
0.735000000000 49.073483806036\\
0.735268420811 49.159748395291\\
0.736000000000 49.172297142618\\
0.737000000000 49.189182624543\\
0.738000000000 49.205763718244\\
0.739000000000 49.222045351620\\
0.740000000000 49.238032378222\\
0.741000000000 49.253729578353\\
0.742000000000 49.269141660156\\
0.743000000000 49.284273260680\\
0.744000000000 49.299128946930\\
0.745000000000 49.313713216905\\
0.746000000000 49.328030500618\\
0.747000000000 49.342085161097\\
0.748000000000 49.355881495382\\
0.749000000000 49.369423735490\\
0.750000000000 49.382716049383\\
0.751000000000 49.395762541907\\
0.752000000000 49.408567255728\\
0.753000000000 49.421134172247\\
0.754000000000 49.433467212501\\
0.755000000000 49.445570238057\\
0.756000000000 49.457447051884\\
0.757000000000 49.469101399221\\
0.758000000000 49.480536968420\\
0.759000000000 49.491757391791\\
0.760000000000 49.502766246422\\
0.761000000000 49.513567054994\\
0.762000000000 49.524163286579\\
0.763000000000 49.534558357429\\
0.764000000000 49.544755631754\\
0.765000000000 49.554758422487\\
0.766000000000 49.564569992035\\
0.767000000000 49.574193553022\\
0.768000000000 49.583632269023\\
0.769000000000 49.592889255282\\
0.770000000000 49.601967579419\\
0.771000000000 49.610870262136\\
0.772000000000 49.619600277898\\
0.773000000000 49.628160555617\\
0.774000000000 49.636553979316\\
0.775000000000 49.644783388791\\
0.776000000000 49.652851580259\\
0.777000000000 49.660761306992\\
0.778000000000 49.668515279955\\
0.779000000000 49.676116168419\\
0.780000000000 49.683566600575\\
0.781000000000 49.690869164136\\
0.782000000000 49.698026406930\\
0.783000000000 49.705040837484\\
0.784000000000 49.711914925598\\
0.785000000000 49.718651102915\\
0.786000000000 49.725251763478\\
0.787000000000 49.731719264282\\
0.788000000000 49.738055925815\\
0.789000000000 49.744264032592\\
0.790000000000 49.750345833684\\
0.791000000000 49.756303543234\\
0.792000000000 49.762139340970\\
0.793000000000 49.767855372709\\
0.794000000000 49.773453750852\\
0.795000000000 49.778936554872\\
0.796000000000 49.784305831799\\
0.797000000000 49.789563596693\\
0.798000000000 49.794711833112\\
0.799000000000 49.799752493571\\
0.800000000000 49.804687500000\\
0.801000000000 49.809518744189\\
0.802000000000 49.814248088228\\
0.803000000000 49.818877364946\\
0.804000000000 49.823408378333\\
0.805000000000 49.827842903967\\
0.806000000000 49.832182689429\\
0.807000000000 49.836429454709\\
0.808000000000 49.840584892614\\
0.809000000000 49.844650669165\\
0.810000000000 49.848628423986\\
0.811000000000 49.852519770695\\
0.812000000000 49.856326297280\\
0.813000000000 49.860049566480\\
0.814000000000 49.863691116147\\
0.815000000000 49.867252459617\\
0.816000000000 49.870735086068\\
0.817000000000 49.874140460869\\
0.818000000000 49.877470025934\\
0.819000000000 49.880725200065\\
0.820000000000 49.883907379286\\
0.821000000000 49.887017937185\\
0.822000000000 49.890058225233\\
0.823000000000 49.893029573116\\
0.824000000000 49.895933289052\\
0.825000000000 49.898770660103\\
0.826000000000 49.901542952490\\
0.827000000000 49.904251411896\\
0.828000000000 49.906897263766\\
0.829000000000 49.909481713608\\
0.830000000000 49.912005947282\\
0.831000000000 49.914471131290\\
0.832000000000 49.916878413062\\
0.833000000000 49.919228921232\\
0.834000000000 49.921523765918\\
0.835000000000 49.923764038992\\
0.836000000000 49.925950814349\\
0.837000000000 49.928085148172\\
0.838000000000 49.930168079191\\
0.839000000000 49.932200628942\\
0.840000000000 49.934183802017\\
0.841000000000 49.936118586317\\
0.842000000000 49.938005953295\\
0.843000000000 49.939846858201\\
0.844000000000 49.941642240317\\
0.845000000000 49.943393023198\\
0.846000000000 49.945100114899\\
0.847000000000 49.946764408204\\
0.848000000000 49.948386780857\\
0.849000000000 49.949968095776\\
0.850000000000 49.951509201279\\
0.851000000000 49.953010931296\\
0.852000000000 49.954474105583\\
0.853000000000 49.955899529931\\
0.854000000000 49.957287996374\\
0.855000000000 49.958640283392\\
0.856000000000 49.959957156110\\
0.857000000000 49.961239366500\\
0.857142857143 49.961419753086\\
0.857234383846 49.998933123242\\
0.858000000000 49.998972512982\\
0.859000000000 49.999022029629\\
0.860000000000 49.999069433292\\
0.861000000000 49.999114803074\\
0.862000000000 49.999158215494\\
0.863000000000 49.999199744564\\
0.864000000000 49.999239461857\\
0.865000000000 49.999277436582\\
0.866000000000 49.999313735644\\
0.867000000000 49.999348423721\\
0.868000000000 49.999381563319\\
0.869000000000 49.999413214843\\
0.870000000000 49.999443436652\\
0.871000000000 49.999472285125\\
0.872000000000 49.999499814713\\
0.873000000000 49.999526078002\\
0.874000000000 49.999551125764\\
0.875000000000 49.999575007012\\
0.876000000000 49.999597769056\\
0.877000000000 49.999619457547\\
0.878000000000 49.999640116535\\
0.879000000000 49.999659788511\\
0.880000000000 49.999678514457\\
0.881000000000 49.999696333893\\
0.882000000000 49.999713284919\\
0.883000000000 49.999729404261\\
0.884000000000 49.999744727313\\
0.885000000000 49.999759288177\\
0.886000000000 49.999773119704\\
0.887000000000 49.999786253532\\
0.888000000000 49.999798720129\\
0.889000000000 49.999810548821\\
0.890000000000 49.999821767838\\
0.891000000000 49.999832404340\\
0.892000000000 49.999842484458\\
0.893000000000 49.999852033324\\
0.894000000000 49.999861075103\\
0.895000000000 49.999869633025\\
0.896000000000 49.999877729417\\
0.897000000000 49.999885385730\\
0.898000000000 49.999892622568\\
0.899000000000 49.999899459719\\
0.900000000000 49.999905916179\\
0.901000000000 49.999912010181\\
0.902000000000 49.999917759218\\
0.903000000000 49.999923180070\\
0.904000000000 49.999928288828\\
0.905000000000 49.999933100918\\
0.906000000000 49.999937631121\\
0.907000000000 49.999941893599\\
0.908000000000 49.999945901917\\
0.909000000000 49.999949669058\\
0.910000000000 49.999953207451\\
0.911000000000 49.999956528985\\
0.912000000000 49.999959645032\\
0.913000000000 49.999962566464\\
0.914000000000 49.999965303670\\
0.915000000000 49.999967866577\\
0.916000000000 49.999970264662\\
0.917000000000 49.999972506974\\
0.918000000000 49.999974602145\\
0.919000000000 49.999976558411\\
0.920000000000 49.999978383620\\
0.921000000000 49.999980085254\\
0.922000000000 49.999981670438\\
0.923000000000 49.999983145955\\
0.924000000000 49.999984518261\\
0.925000000000 49.999985793497\\
0.926000000000 49.999986977499\\
0.927000000000 49.999988075813\\
0.928000000000 49.999989093708\\
0.929000000000 49.999990036181\\
0.930000000000 49.999990907976\\
0.931000000000 49.999991713589\\
0.932000000000 49.999992457280\\
0.933000000000 49.999993143082\\
0.934000000000 49.999993774814\\
0.935000000000 49.999994356085\\
0.936000000000 49.999994890308\\
0.937000000000 49.999995380704\\
0.938000000000 49.999995830316\\
0.939000000000 49.999996242011\\
0.940000000000 49.999996618492\\
0.941000000000 49.999996962304\\
0.942000000000 49.999997275840\\
0.943000000000 49.999997561353\\
0.944000000000 49.999997820955\\
0.945000000000 49.999998056631\\
0.946000000000 49.999998270240\\
0.947000000000 49.999998463526\\
0.948000000000 49.999998638116\\
0.949000000000 49.999998795536\\
0.950000000000 49.999998937208\\
0.951000000000 49.999999064458\\
0.952000000000 49.999999178522\\
0.953000000000 49.999999280551\\
0.954000000000 49.999999371614\\
0.955000000000 49.999999452701\\
0.956000000000 49.999999524732\\
0.957000000000 49.999999588558\\
0.958000000000 49.999999644964\\
0.959000000000 49.999999694676\\
0.960000000000 49.999999738361\\
0.961000000000 49.999999776634\\
0.962000000000 49.999999810058\\
0.963000000000 49.999999839149\\
0.964000000000 49.999999864380\\
0.965000000000 49.999999886181\\
0.966000000000 49.999999904944\\
0.967000000000 49.999999921024\\
0.968000000000 49.999999934744\\
0.969000000000 49.999999946396\\
0.970000000000 49.999999956241\\
0.971000000000 49.999999964515\\
0.972000000000 49.999999971429\\
0.973000000000 49.999999977172\\
0.974000000000 49.999999981909\\
0.975000000000 49.999999985790\\
0.976000000000 49.999999988946\\
0.977000000000 49.999999991489\\
0.978000000000 49.999999993522\\
0.979000000000 49.999999995129\\
0.980000000000 49.999999996388\\
0.981000000000 49.999999997361\\
0.982000000000 49.999999998104\\
0.983000000000 49.999999998662\\
0.984000000000 49.999999999076\\
0.985000000000 49.999999999376\\
0.986000000000 49.999999999590\\
0.987000000000 49.999999999739\\
0.988000000000 49.999999999839\\
0.989000000000 49.999999999905\\
0.990000000000 49.999999999947\\
0.991000000000 49.999999999972\\
0.992000000000 49.999999999986\\
0.993000000000 49.999999999994\\
0.994000000000 49.999999999998\\
0.995000000000 49.999999999999\\
0.996000000000 50.000000000000\\
0.997000000000 50.000000000000\\
0.998000000000 50.000000000000\\
0.999000000000 50.000000000000\\
0.999900000000 50.000000000000\\
};
\addplot[only marks,mark=*,mark size=1.5pt,black,forget plot] coordinates {(.5,0)};
\addplot[only marks,mark=o,mark size=2.1pt,mark options={solid,fill=white},black,forget plot] coordinates {(1,50)};
\end{axis}
\end{tikzpicture}
\caption{Optimal prize rules and the cost of neutrality: seven battles. In panel (a), slight vertical offsets separate coincident curves; all thresholds are integers.}
\label{fig:design-comparison}
\end{figure}

\end{example}

At $p=0.55$, both design classes select simple majority: four victories secure the prize. Neutrality is costless because the most likely pivotal count in the other six battles is three. Simple majority remains optimal in both classes up to $p=4/7$, that is, for cost ratios up to $(4/3)^7\approx7.49$ when $r=1/7$. At $p=0.65$, it remains optimal under neutrality, whereas the optimal identity-dependent rule requires five victories from $A$ but only three from $B$. Neutrality can thus reduce effort before its optimal threshold changes, retaining $28/39$ (approximately $71.8\%$) of optimal identity-dependent effort.

At $p=0.70$, both rules require five victories from $A$ for the full prize. If $A$ wins three or four battles, however, the neutral rule splits the prize, whereas the identity-dependent rule awards it entirely to $B$. Neutrality retains $29/49$ (approximately $59.2\%$) of optimal identity-dependent effort: equal full-prize thresholds need not imply equal incentives.

The optimal neutral threshold increases at approximately $p=0.689$, $0.735$, and $0.857$ (slightly above $6/7$), and the identity-dependent threshold at $4/7$, $5/7$, and $6/7$; at each switch, both adjacent threshold rules are optimal, and relative effort remains continuous. The $50\%$ loss at $p=1$ is a limiting value. Because $c_B/c_A=[p/(1-p)]^{1/r}$, high winning probabilities imply substantial cost asymmetry under low discriminatory power: at $r=1/7$, $p=0.75$ corresponds to $c_B/c_A=2{,}187$.

\section{Concluding remarks}\label{Sec: Conclusion}

Within sharp bounds on discriminatory power that guarantee a unique pure equilibrium for every feasible rule and cost ratio, an optimal neutral rule splits the prize unless one player wins by a sufficient margin, an optimal identity-dependent rule gives the weaker player a possibly zero headstart, and the proportional effort cost of neutrality weakly increases with asymmetry. Appendix~\ref{app:section4-proofs} proves \autoref{Lemma for p} and \autoref{prop:relative-cost}; equilibrium proofs are in \OAref{app:section3-proofs}.

We consider power-form convex costs in \OAref{app:costs}. A limitation of our analysis is that the equilibrium guarantees, and hence the prize-design characterizations, require increasingly weak discrimination as the number of battles grows. Extending optimal prize design beyond our pure-equilibrium guarantees remains a direction for future research. A starting point is the equilibrium characterization of \cite{KuangLu2026Architecture}, which allows mixed strategies under identity-independent rules for $0<r\leq1$. Insights from one-dimensional contests may also be useful. \cite{Wang2010optimal} studies the optimal choice of discriminatory power, while \cite{Ewerhart2015Mixed} and \cite{Ewerhart2017Noise} develop methods for analyzing mixed equilibria in asymmetric Tullock contests and all-pay auctions.

\appendix

\section{Proofs for \autoref{Sec: Prize Design}}\label{app:proofs}\label{app:section4-proofs}

Here we prove \autoref{Lemma for p} and \autoref{prop:relative-cost}. If $n=1$, both design classes have the sole optimal threshold $1$ and $\rho\equiv1$. We therefore assume $n\geq2$ below.

Write $m=n-1$. Recall that $\theta(t\mid m,p)=\binom mt p^t(1-p)^{m-t}$ is the probability of $t$ victories in $m$ independent battles. A \emph{mode} is a count that maximizes this probability. Under identity independence, we instead maximize the sum $g(k,p)=\theta(k\mid m,p)+\theta(m-k\mid m,p)$ at the two reflected counts $k$ and $m-k$.

\subsection{Proof of \autoref{Lemma for p}}\label{app:statics-proof}

For parts (i) and (ii), we show that the incentive value of a higher threshold relative to a lower one strictly increases with the winning probability. This rules out a downward change in any optimal threshold. For part (iii), we compare a maximizing reflected pair with a binomial mode.

\noindent\textit{(i) Identity-independent rules.} Let $m=n-1$, $p=\bar p_A\in[1/2,1)$, and $z=\log[p/(1-p)]\geq0$. By \autoref{thm:tiemargin}, an optimal threshold is $T=k+1$, where $k$ maximizes $g(k,p)$ over $\{\lceil m/2\rceil,\ldots,m\}$.
The log-odds $z$ increase with $p$. They let us separate a factor common to all thresholds from a factor depending on the distance $k-m/2$ from the central count. Using $\binom m{m-k}=\binom mk$ and $\cosh y=(e^y+e^{-y})/2$,
\begin{align*}
g(k,p)
&=\binom mk
  \left[p^k(1-p)^{m-k}
       +p^{m-k}(1-p)^k\right]\\
&=2\binom mk[p(1-p)]^{m/2}
  \cosh\!\big((k-m/2)z\big).
\end{align*}
When $m$ is even and $k=m/2$, the two reflected counts coincide. The formula still applies: $g(m/2,p)=2\theta(m/2\mid m,p)$, so the expected marginal prize $g/2$ equals the probability at that single count.

Fix two feasible upper indices $a<b$, and let $\alpha_a=a-m/2$ and $\alpha_b=b-m/2$, so that $\alpha_b>\alpha_a\geq0$. The common factor cancels:
\[
\frac{g(b,p)}{g(a,p)}
=\frac{\binom mb}{\binom ma}
 \frac{\cosh(\alpha_bz)}{\cosh(\alpha_az)}.
\]
Because the ratio is positive, its logarithm increases exactly when the ratio increases. Using $\tanh y=(e^y-e^{-y})/(e^y+e^{-y})$, its logarithmic derivative with respect to $z$ is
\[
\frac{d}{dz}\log\frac{g(b,p)}{g(a,p)}
=\alpha_b\tanh(\alpha_bz)-\alpha_a\tanh(\alpha_az).
\]
For $z>0$,
\[
\frac{\partial}{\partial\alpha}
 \big[\alpha\tanh(\alpha z)\big]
=\tanh(\alpha z)+\frac{\alpha z}{\cosh^2(\alpha z)}
>0
\qquad(\alpha>0).
\]
Since $\alpha\tanh(\alpha z)=0$ at $\alpha=0$, this function is strictly increasing on $\alpha\geq0$.
Therefore, the logarithmic derivative is positive for every $z>0$. For any $0\leq z_1<z_2$, integrating this derivative over $[z_1,z_2]$ gives a strictly positive change in the log ratio, including when $z_1=0$.
Since $\frac{dz}{dp}=\frac{1}{p(1-p)}>0$, $g(b,p)/g(a,p)$ strictly increases between any two probabilities $p_2>p_1\geq1/2$.

Now take any maximizing indices $k_1$ at $p_1$ and $k_2$ at $p_2>p_1$.
Suppose, contrary to the claim, that $k_2<k_1$. Optimality at $p_1$ and the strict ratio comparison imply
\[
\frac{g(k_1,p_2)}{g(k_2,p_2)}
>
\frac{g(k_1,p_1)}{g(k_2,p_1)}
\geq1,
\]
contradicting optimality of $k_2$ at $p_2$.
Hence $T_2=k_2+1\geq k_1+1=T_1$. Since the maximizers were arbitrary, this ordering includes every choice among ties.

\noindent\textit{(ii) Identity-dependent rules.} By \autoref{pro:headstart}, an optimal headstart threshold is $\hat T=k+1$, where $k$ maximizes $\theta(k\mid m,p)$ over $\{0,\ldots,m\}$. For $0\leq a<b\leq m$,
\[
\frac{\theta(b\mid m,p)}{\theta(a\mid m,p)}
=\frac{\binom mb}{\binom ma}
 \left(\frac{p}{1-p}\right)^{b-a},
\]
which strictly increases with $p\in(0,1)$. The same arbitrary-maximizer argument as in part (i) gives $\hat T_1\leq\hat T_2$, including $p_1=1/2$ and any choice among tied binomial modes.

\noindent\textit{(iii) Comparing the two designs.} Let $p>1/2$. For integers $0\leq t<m/2$,
\[
\frac{\theta(t+1\mid m,p)}{\theta(t\mid m,p)}
=\frac{m-t}{t+1}\frac p{1-p}>1.
\]
Thus binomial probabilities strictly increase along the lower half, and every mode $j$ satisfies $j\geq\lceil m/2\rceil$. Fix any such mode. If $k>j$, then $\theta(k\mid m,p)\leq\theta(j\mid m,p)$ because $j$ is a mode. Also, $m-k<m-j\leq m/2$, so the lower-half monotonicity gives $\theta(m-k\mid m,p)<\theta(m-j\mid m,p)$. Adding these inequalities yields $g(k,p)<g(j,p)$. Moving above $j$ therefore cannot improve the upper-count probability and strictly lowers the reflected-count probability. No maximizing upper index exceeds $j$, so $T\leq j+1=\hat T$ for every pair of optimal thresholds. At $p=1/2$, the binomial distribution has a unique central mode when $n$ is odd and two adjacent central modes when $n$ is even. The maximizing reflected pair is central in either case, giving the stated symmetry cases.

\subsection{Proof of \autoref{prop:relative-cost}}\label{app:relative-cost-proof}

Fix $p\in(1/2,1)$ and let $m=n-1$, $\theta_k=\theta(k\mid m,p)$, and $d=\lfloor np\rfloor$, the largest binomial mode. To see the mode formula, for every $0\leq t<m$ the same adjacent-probability calculation as in part (iii) gives $\theta(t+1\mid m,p)/\theta(t\mid m,p)\geq1$ exactly when $t+1\leq np$. Thus probabilities stop increasing at $d$; when $np$ is an integer, both $d-1$ and $d$ are modes. Represent a maximizing reflected pair by its upper index $k\in\{\lceil m/2\rceil,\ldots,m\}$. Part (iii) of \autoref{Lemma for p} gives $k\leq d$.

We first consider an interval on which the maximizing indices $k$ and $d$ do not change. Such intervals cover the domain except for finitely many points: every candidate probability or reflected-pair sum is a polynomial in $p$, and two distinct polynomials coincide only finitely often. Identical polynomials can be treated as a single candidate.

Let $q=p/(1-p)>1$. Cancelling the common factor $(1-p)^m$ from the binomial probabilities gives
\[
\rho(p)=\frac{\theta_k+\theta_{m-k}}{2\theta_d}
=\frac{\binom mk}{2\binom md}
  \left(q^{k-d}+q^{m-k-d}\right).
\]
Since $m-k\leq k\leq d$, both exponents are nonpositive. Each power of $q$ therefore weakly decreases as $q$, and hence $p$, rises. This proves that $\rho$ is weakly decreasing on each interval with fixed maximizing indices. A maximum of finitely many continuous functions is continuous, and the denominator is positive. Thus $\rho$ has no jump when a maximizing index changes or at $p=1/2$. The interval-by-interval conclusion therefore extends to all of $[1/2,1)$, and $1-\rho$ is weakly increasing.

\begin{samepage}
For the bounds, each pair average $(\theta_k+\theta_{m-k})/2$ is at most the largest binomial probability. A pair containing a mode has average strictly greater than half that probability, since the reflected probability is positive for $p<1$. Dividing by the largest probability gives $1/2<\rho\leq1$. At $p=1/2$, a central pair attains the largest probability, giving $\rho(1/2)=1$. As $p\uparrow1$, only $\theta_m$ tends to one and all other probabilities tend to zero. The best pair average therefore tends to $1/2$, while the largest probability tends to one, so $\rho(p)\to1/2$.
\par\end{samepage}

Finally, $\rho(p)=1$ means that a reflected-pair average equals the largest binomial probability, so both probabilities must be maximal. For $p>1/2$, however, $\theta_k/\theta_{m-k}=[p/(1-p)]^{2k-m}>1$ whenever $k>m/2$. Equality therefore requires the central count $k=m/2$, which exists only for odd $n$. This count is a mode exactly when $m/2\leq np\leq m/2+1$, or $p\leq(n+1)/(2n)$ on $p\geq1/2$. Together with $\rho(1/2)=1$, this establishes both equality cases.

\section*{Declaration of generative AI and AI-assisted technologies in the manuscript preparation process}

While preparing this manuscript, the authors used ChatGPT and Claude to improve the language and assist with checking mathematical arguments. The authors independently reviewed and revised the assisted material and accept full responsibility for the article's contents.

\begingroup
\setstretch{1.1}
\ifx\undefined\bysame
\newcommand{\bysame}{\leavevmode\hbox to\leftmargin{\hrulefill\,\,}}
\fi

\endgroup

\clearpage
\pagenumbering{arabic}
\renewcommand{\thepage}{OA-\arabic{page}}
\setcounter{lemma}{0}
\setcounter{equation}{0}
\numberwithin{equation}{section}
\renewcommand{\thelemma}{\thesection.\arabic{lemma}}
\setlength{\emergencystretch}{1.5em}
\normalfont
\pdfbookmark[0]{Online Appendix}{combined-online-appendix}

\begin{center}
{\Large\bfseries Online Appendix}\par\medskip
{\large Generalized Multidimensional Contests with Asymmetric Players:\\ Equilibrium and Optimal Prize Design}\par\medskip
Siyuan Fan, Zhonghong Kuang, and Jingfeng Lu\par
October 2026
\end{center}
\noindent This supplement contains the detailed equilibrium proofs (Appendix~A) and the power-cost extensions (Appendix~B). References to numbered results and sections in the article retain the article's numbering; auxiliary lemmas and equations are numbered within this supplement.
\appendix
\section{Equilibrium proofs}\label{OA-app:section3-proofs}

Here we prove \autoref{lem:Uniform Strategy}, \autoref{Theo:existence}, \autoref{lem:pure-uniformity}, \autoref{prop:necessary}, and \autoref{rem:ID-equilibrium}.

\subsection{Proof of \autoref{lem:Uniform Strategy}}\label{OA-app:uniform-proof}

We use the following interchangeability property to prove \autoref{lem:Uniform Strategy}. Given a two-player game, we say that two equilibria $(\mu_A,\mu_B)$ and $(\mu_A^{'},\mu_B^{'})$ are \textbf{interchangeable} if $(\mu_A^{'},\mu_B)$ and $(\mu_A,\mu_B^{'})$ are equilibria as well.

\begin{lemma}\label{OA-lem:interchangeable}
Suppose expected prizes and costs are finite and the players' payoff functions satisfy\\
1. Own-strategy costs: $\pi_i(\mu_A,\mu_B) = F_i(\mu_A,\mu_B) - C_i(\mu_i)$,\\
2. Fixed sum: $F_A(\mu_A,\mu_B)+F_B(\mu_A,\mu_B)$ is a constant.\\
The equilibria of a two-player game are \textbf{interchangeable}.
\end{lemma}

\begin{proof}
Following \cite{Klumpp2006Primaries} and \cite{Ewerhart2017Revenue}, define $u_i=\pi_i+C_{-i}$. This preserves best responses and makes $u_A+u_B=F_A+F_B$ constant. Equilibrium strategies in a constant-sum game are interchangeable, proving the claim.

In our contest, restricting to finite expected costs loses no equilibrium or profitable deviation: prizes lie in $[0,1]$, so infinite expected cost yields payoff $-\infty$, whereas zero effort guarantees a nonnegative payoff.
\end{proof}

\paragraph{Part (i): Lifting uniform equilibria to the full game.}

A constant feasible prize rule pays each player $1/2$ regardless of effort, so zero effort is the unique equilibrium in both games. For $n=1$, every strategy is uniform and the two games coincide. For the remainder of the proof, let $n\geq2$, $0<r\leq1$, and $v$ be nonconstant.

\textbf{Step 1}. Fix player $A$'s \emph{positive uniform effort $x_A$} and player $B$'s nonuniform vector $\bm x_B$. Choose battles $k,l$ with $x_{B(k)}>x_{B(l)}\geq0$. We first show that averaging these two efforts weakly improves $B$'s payoff. We establish strictness \emph{when $x_{B(j)}>0$ for every $j\in\{1,\ldots,n\}\setminus\{k,l\}$}.

Keep all other components unchanged and set
\begin{equation*}
\tilde{x}_{B(k)}=\tilde{x}_{B(l)}=\frac{x_{B(k)}+x_{B(l)}}{2}\triangleq\tilde{x}_{B}.
\end{equation*}
This preserves total effort and cost.

For $s\in\{0,1,2\}$, let $Q_B(s)$ and $\tilde{Q}_B(s)$ denote the probabilities that player $B$ wins exactly $s$ of battles $k$ and $l$ before and after averaging, respectively. Let $Q^\neg_B(t)$, $t\in\{0,1,\ldots,n-2\}$, denote the probability of winning exactly $t$ of the remaining $n-2$ battles.

Hence, player $B$'s gain from changing to the new strategy is
\begin{eqnarray}\label{OA-Equa:gain}
   && \Delta\pi_B\\ \nonumber
   && =\sum_{t=0}^{n-2}Q^\neg _B(t)\left(v(t)(\tilde{Q}_B(0)-Q_B(0))+v(t+1)(\tilde{Q}_B(1)-Q_B(1))+v(t+2)(\tilde{Q}_B(2)-Q_B(2))\right).
\end{eqnarray}
Since $\{Q^\neg _B(t)\}_{t=0}^{n-2}$ do not change when switching from strategy $\bm x_B$ to $\tilde{\bm x}_B$, a sufficient condition for $\Delta\pi_B$ to be nonnegative is that
\begin{align*}
&v(t)(\tilde{Q}_B(0)-Q_B(0))+v(t+1)(\tilde{Q}_B(1)-Q_B(1))+v(t+2)(\tilde{Q}_B(2)-Q_B(2))\\
=&v(t)(\tilde{Q}_B(0)-Q_B(0))+v(t+1)\left([1-\tilde{Q}_B(2)-\tilde{Q}_B(0)]-[1-Q_B(0)-Q_B(2)]\right)\\
&\ \ \ +v(t+2)(\tilde{Q}_B(2)-Q_B(2))\\
=&(v(t)-v(t+1))(\tilde{Q}_B(0)-Q_B(0))+(v(t+2)-v(t+1))(\tilde{Q}_B(2)-Q_B(2))\geq 0.
\end{align*}
The next two calculations establish $\tilde{Q}_B(0)<Q_B(0)$ and $\tilde{Q}_B(2)>Q_B(2)$. Together with monotonicity of $v$, these inequalities imply $\Delta\pi_B\geq0$ in \autoref{OA-Equa:gain}. It is strictly positive if \emph{the remaining $n-2$ components are positive} and $v$ is nonconstant: $\{Q^\neg_B(t)\}_{t=0}^{n-2}$ then has full support, and every increment $v(t+1)-v(t)$, $t=0,\ldots,n-1$, appears in \autoref{OA-Equa:gain} with a strictly positive coefficient. With zero components elsewhere, strictness need not hold.

\textbf{I. $\tilde{Q}_B(0)<Q_B(0)$}.
\begin{align*}
\tilde{Q}_B(0)-Q_B(0)
&= \frac{x_A^r}{x_A^r+\tilde{x}_B^r}\cdot \frac{x_A^r}{x_A^r+\tilde{x}_B^r}-\frac{x_A^r}{x_A^r+x_{B(k)}^r}\cdot \frac{x_A^r}{x_A^r+x_{B(l)}^r} \\
&= x_A^{2r}\cdot \frac{(x_A^r+x_{B(k)}^r)(x_A^r+x_{B(l)}^r)-(x_A^r+\tilde{x}_B^r)^2}{(x_A^r+x_{B(k)}^r)(x_A^r+x_{B(l)}^r)(x_A^r+\tilde{x}_B^r)^2}\\
&= x_A^{2r}\cdot \frac{x_A^r\big(x_{B(k)}^r+x_{B(l)}^r-2\tilde{x}_B^r\big)+\big(x_{B(k)}^rx_{B(l)}^r-\tilde{x}_B^{2r}\big)}{(x_A^r+x_{B(k)}^r)(x_A^r+x_{B(l)}^r)(x_A^r+\tilde{x}_B^r)^2}.
\end{align*}
Recall that $\tilde{x}_B=\frac{x_{B(k)}+x_{B(l)}}{2}$ and $r\in(0,1]$. Since $x^r$ is a concave function, we have $x_{B(k)}^r+x_{B(l)}^r\leq2\tilde{x}_B^r$. Meanwhile, $\tilde{x}_B>\sqrt{x_{B(k)}x_{B(l)}}$ by the strict arithmetic--geometric mean inequality, so we have $x_{B(k)}^rx_{B(l)}^r<\tilde{x}_B^{2r}$.

\textbf{II. $\tilde{Q}_B(2)>Q_B(2)$}.
\begin{align*}
\tilde{Q}_B(2)-Q_B(2)
&=\frac{\tilde{x}_B^r}{x_A^r+\tilde{x}_B^r}\cdot \frac{\tilde{x}_B^r}{x_A^r+\tilde{x}_B^r} - \frac{x_{B(k)}^r}{x_A^r+x_{B(k)}^r}\cdot \frac{x_{B(l)}^r}{x_A^r+x_{B(l)}^r}\\
&=\frac{\tilde{x}_{B}^{2r}(x_A^r+x_{B(k)}^r)(x_A^r+x_{B(l)}^r)-x_{B(k)}^r x_{B(l)}^r(x_A^r+\tilde{x}_{B}^r)^2}{(x_A^r+x_{B(k)}^r)(x_A^r+x_{B(l)}^r)(x_A^r+\tilde{x}_{B}^r)^2}\\
&= \frac{x_A^{2r}\big(\tilde{x}_{B}^{2r}-x_{B(k)}^r x_{B(l)}^r\big)
\;+\;x_A^{r}\tilde{x}_{B}^{r}\Big(\tilde{x}_{B}^{r}\big(x_{B(k)}^r+x_{B(l)}^r\big)-2x_{B(k)}^r x_{B(l)}^r\Big)}{(x_A^r+x_{B(k)}^r)(x_A^r+x_{B(l)}^r)(x_A^r+\tilde{x}_{B}^r)^2}.
\end{align*}
The second term in the numerator is positive because
\begin{equation*}
\tilde{x}_{B}^{r} \ge \frac{x_{B(k)}^r+x_{B(l)}^r}{2} > \sqrt{x_{B(k)}^r x_{B(l)}^r}.
\end{equation*}
The last inequality is strict because $x_{B(k)}>x_{B(l)}$.

\textbf{Step 2}. Against an opponent using \emph{positive uniform effort}, the uniform strategy $\bar{\bm{x}}_B$ with the same total expenditure yields strictly higher payoff than any nonuniform strategy $\bm{x}_B$. To show this, repeatedly average a largest and a smallest component. While zeros remain, each step removes one zero, since the maximum is positive. All components therefore become positive in finitely many steps. By Step 1, averaging unequal components is strict whenever all other components are positive, including at the step that removes the last zero. Hence a strict gain occurs in finitely many steps.

The process converges to the uniform vector. If $\bar x$ is the fixed mean and $D=\sum_j(x_{B(j)}-\bar x)^2$, averaging the maximum $a$ and minimum $b$ reduces $D$ by $(a-b)^2/2\geq D/(2n)$. Thus $D$ tends to zero. Against an opponent using \emph{positive uniform effort}, expected payoff is continuous. Weak improvement at every step and a strict improvement at a finite step give a strict improvement at the limit.

\textbf{Step 3}. If player $A$ chooses zero in all battles, complete averaging weakly improves player $B$'s expected prize. When $B$'s total effort is positive, averaging leaves existing sure wins unchanged and converts any remaining battles into sure wins, at the same cost; the zero vector is unchanged. For a mixed deviation, replacing each realization $\bm x_B$ by its complete average $\bar{\bm x}_B$ weakly improves payoff against every uniform realization of player $A$'s strategy. Thus every mixed deviation has a mixed uniform counterpart with weakly higher expected payoff. This proves part (i).

\paragraph{Part (ii): Excluding nonuniform equilibria.}

In a restricted equilibrium, neither player can choose zero almost surely. Against an identically zero opponent, any response assigning positive probability to positive effort is strictly improved by halving every positive effort realization: outcomes are unchanged and finite expected cost falls strictly. An all-zero response is inferior to a sufficiently small positive uniform effort under a nonconstant rule. Thus every restricted-equilibrium strategy assigns positive probability to positive scalar effort. Consequently, any strategy assigning positive probability to nonuniform vectors is strictly improved by averaging those vectors: the gain is strict against \emph{positive uniform opponent realizations} and weak against zero. Finite expected costs make these comparisons integrable.

Budget balance gives $F_A+F_B=1$, and each player's cost depends only on that player's strategy. Hence \autoref{OA-lem:interchangeable} applies to $\mathcal G$. Suppose $(\mu_A,\mu_B)$ is a uniform equilibrium and $(\mu_A',\mu_B')$ is another equilibrium in which $\mu_B'$ assigns positive probability to nonuniform vectors. Interchangeability makes $(\mu_A,\mu_B')$ an equilibrium as well. But complete averaging strictly improves $\mu_B'$ against the uniform opponent $\mu_A$, contradicting best response. Therefore, if the restricted game has an equilibrium, every equilibrium of the full game is uniform.

\subsection{Proof of \autoref{Theo:existence}}\label{OA-app:existence-proof}

\begin{lemma}\label{OA-lem:V'V}
For a nonconstant identity-independent rule and every $p\in(0,1)$,
\[
L(p)\equiv(1-2p)+p(1-p)\frac{\mathcal V'(p)}{\mathcal V(p)}<\lceil n/2\rceil.
\]
\end{lemma}

\begin{proof}
First note that 
\begin{align*}
\frac{\partial \theta(t\mid n-1,p)}{\partial p}&=t\binom{n-1}{t}p^{t-1}(1-p)^{n-t-1}-(n-t-1)\binom{n-1}{t}p^t(1-p)^{n-t-2}\\
&=\theta(t\mid n-1,p)\left(\frac{t}{p}-\frac{n-t-1}{1-p}\right)=\theta(t\mid n-1,p)\frac{t-p(n-1)}{p(1-p)}.
\end{align*}
Hence, we obtain
\begin{equation}
p(1-p)\mathcal V'(p)=\sum_{t=0}^{n-1}\Delta v(t)\theta(t\mid n-1,p)\big(t-(n-1)p\big).
\end{equation}
Let $m=n-1$ and $\omega_t=\Delta v(t)\theta(t\mid m,p)/\mathcal V(p)$, so that $\omega_t\geq0$ and $\sum_{t=0}^{m}\omega_t=1$. Dividing the preceding identity by $\mathcal V(p)>0$ yields
\begin{equation}\label{OA-eq:lemmaV(p)}
p(1-p)\frac{\mathcal V'(p)}{\mathcal V(p)}=\sum_t t\omega_t-mp,
\qquad L(p)=1+\sum_t t\omega_t-(n+1)p.
\end{equation}
If $p>1/2$, then $\sum_t t\omega_t\leq n-1$, so $L(p)\leq n-(n+1)p<(n-1)/2$.

For $p\leq1/2$, use $\Delta v(t)=\Delta v(m-t)$. In each reflected pair the lower index has at least as much binomial weight as the upper one, so the weighted mean is at most $m/2$. For odd $n$, this gives
\[
L(p)\leq(n+1)/2-(n+1)p<(n+1)/2.
\]
This argument includes the central singleton and $n=1$.

For even $n=2h$, pair $t=h-1-a$ and $t'=h+a$, where $a=0,\ldots,h-1$. Set $z=p/(1-p)\leq1$. In a pair with positive marginal weight, the conditional weight on $t'$ is
\[
\frac{\theta(t'\mid n-1,p)}{\theta(t\mid n-1,p)+\theta(t'\mid n-1,p)}
=\frac{z^{2a+1}}{1+z^{2a+1}}\leq\frac{z}{1+z}=p.
\]
Here $0<z\leq1$ implies $z^{2a+1}\leq z$. The pair's mean index is therefore at most $h-1-a+(2a+1)p=h-1+p+a(2p-1)\leq h-1+p$. Averaging over pairs gives $L(p)\leq h-np<h$. Pairs with zero marginal weight contribute nothing. The final inequality is strict because $p>0$; the within-pair inequality need not be strict.
\end{proof}

\paragraph{Sufficiency.} The constant rule induces zero effort. For a nonconstant rule, normalization preserves equilibria up to effort rescaling by the argument in \autoref{Sec: Equilibrium Analysis}, so suppose $v(0)=0$ and $v(n)=1$. We consider player $B$'s problem by fixing $x_A^*>0$. We show that $\pi_B(x_B\mid x_A^*)$ increases up to $x_B^*$ and decreases above it.

Define the binomial tail $H(k\mid m,p)=\sum_{t=k}^m\theta(t\mid m,p)$. Write $u=\ln x_B$ and treat $B$'s per-battle winning probability as a function of $u$: $p(u)\triangleq\frac{e^{ru}}{e^{ru}+(x_A^*)^r}\in(0,1)$ and $\partial p/\partial u=rp(1-p)$. Then, player $B$'s payoff can be rewritten as
\[
\pi_B(u\mid x_A^*)=\sum_{t=0}^{n-1}\Delta v(t)H(t+1\mid n,p(u))-nc_Be^u.
\]
Differentiating with respect to $u$ gives
\[
\pi_B'(u\mid x_A^*)=nrp(1-p)\mathcal V(p)-nc_Be^u,
\]
using $\partial H(k\mid n,p)/\partial p=n\theta(k-1\mid n-1,p)$.

To establish that $\pi_B(u)$ is single-peaked, it suffices to show that its first derivative $\pi_B'(u)$ \emph{crosses zero exactly once and from above}. For this argument, we are only concerned with the \emph{sign} of $\pi_B'(u)$, not its magnitude. Inspired by \cite{zenou2024sign}, we define an auxiliary function that is \emph{sign-equivalent} to $\pi_B'(u\mid x_A^*)$:
\[
G(u\mid x_A^*)\triangleq\ln\!\big(rp(1-p)\mathcal V(p)\big)-u-\ln c_B.
\]
Since $\mathcal V(p)>0$, $G$ and $\pi_B'$ have the same sign.
By \autoref{OA-lem:V'V},
\begin{equation}\label{OA-eq:Fprime_app3}
G'(u\mid x_A^*)=rL(p(u))-1<r\lceil n/2\rceil-1\leq0.
\end{equation}
Since $G(u\mid x_A^*)$ is \emph{strictly decreasing} and vanishes at $u=\ln x_B^*$, both $G$ and $\pi_B'$ are positive below this point and negative above it. This ensures that $x_B^*$ is the \emph{unique} best response to $x_A^*$ among positive uniform efforts. Strict increase in payoff over $0<x_B<x_B^*$, together with continuity at zero, excludes zero as another maximizer. By symmetry of the argument, we can verify that $x_A^*$ is the unique best response to $x_B^*$. Thus the candidate is an equilibrium of the restricted game $\widetilde{\mathcal G}$. Since $r\leq b_n\leq1$, \autoref{lem:Uniform Strategy} extends it to the full game $\mathcal G$. This finishes the proof of \autoref{Theo:existence}.

\subsection{Proof of \autoref{lem:pure-uniformity}}\label{OA-app:pure-uniformity-proof}

The argument uses only necessary conditions for a pure equilibrium and does not rely on \autoref{lem:Uniform Strategy}; it therefore applies for every $r>0$.

\begingroup
\PreventProofWidow
\begin{proof}
In a pure equilibrium, one-sided positive effort in a battle could be reduced without changing its outcome. Each battle therefore has either two positive efforts or two zero efforts. Conditional independence then gives the number $K_{-j}$ of $A$'s wins outside battle $j$ full support on $\{0,\ldots,n-1\}$. Since $v$ is nonconstant and nondecreasing, the marginal prize $D_j=\mathbb E[v(K_{-j}+1)-v(K_{-j})]$ is positive. If both efforts in battle $j$ were zero, a sufficiently small positive effort would yield the fixed prize gain $D_j/2$ at arbitrarily small cost. Hence all efforts are positive.

Budget balance gives both players the same marginal prize $D_j$. Holding all other efforts fixed, the marginal expected prize from effort in battle $j$ is $D_j\,\partial p_{i(j)}/\partial x_{i(j)}$. Differentiating the Tullock success function therefore gives the first-order condition
\[
\frac{\partial\pi_i}{\partial x_{i(j)}}
=\frac{r p_{i(j)}(1-p_{i(j)})D_j}{x_{i(j)}}-c_i=0,
\qquad i=A,B.
\]
Multiplying by $x_{i(j)}$ and using $p_{B(j)}=1-p_{A(j)}$ yields
\[
c_Ax_{A(j)}=rp_{A(j)}(1-p_{A(j)})D_j=c_Bx_{B(j)}.
\]
The effort ratio is therefore $c_B/c_A$ in every battle, yielding winning probabilities $\bar p_A$ and $\bar p_B$. Conditional independence makes $K_{-j}$ binomial with parameters $(n-1,\bar p_A)$, so $D_j=\mathcal V$ in every battle. Substitution gives the stated uniform efforts.
\end{proof}
\endgroup

\subsection{Proof of \autoref{prop:necessary}}\label{OA-app:sharpness-proof}

\begin{proof}
Let $h=\lceil n/2\rceil$ and $q=\bar p_B$. As $q\downarrow0$, the majority-rule candidate payoff satisfies $\pi_B^*=A_nq^h[1-rh+O(q)]$, where $A_n=\binom nh$ for odd $n$ and $A_n=\tfrac12\binom nh$ for even $n$.
For fixed $n$ and $r$, $O(q)$ is bounded in absolute value by a constant times $q$ as $q\downarrow0$. To verify the cost term, write $W_B(q)$ for the expected prize under majority rule, with an equal split at a tie. The candidate's total effort cost is $rq(1-q)W_B'(q)$, because $W_B'(q)=n\mathcal V(q)$. Expected prize and total effort cost therefore have leading terms $A_nq^h$ and $rhA_nq^h$, respectively. For even $n$, the factor $1/2$ reflects the half-prize at a tie. Since $q\downarrow0$ as $c_B/c_A\to\infty$ and $r>b_n=1/h$, the candidate payoff is negative for sufficiently large cost ratios. Player $B$ would profitably withdraw for zero payoff, ruling out the candidate. \autoref{lem:pure-uniformity} also rules out nonuniform pure equilibria. Hence no pure-strategy equilibrium exists.
\end{proof}

\subsection{Proof of \autoref{rem:ID-equilibrium}}\label{OA-app:ID-proof}

\noindent\textit{Part (i).} We follow the proofs of \autoref{Theo:existence} (\OAref{app:existence-proof}) and \autoref{pro:unique} (\autoref{subsec:unique}): first bound the prize-spread elasticity, then establish existence by the sign-equivalent auxiliary function, and finally prove uniqueness.

\paragraph{Prize-spread bound.}
For a nonconstant identity-dependent rule and $i\in\{A,B\}$, write
\[
\Delta v^i(t)=v^i(t+1)-v^i(t),\qquad
V_i(p)=\sum_{t=0}^{n-1}\Delta v^i(t)\theta(t\mid n-1,p).
\]
Budget balance gives $v^B(k)=1-v^A(n-k)$ and hence $\Delta v^B(t)=\Delta v^A(n-1-t).$ Using $\theta(t\mid n-1,1-p)=\theta(n-1-t\mid n-1,p)$, we obtain
\begin{align*}
V_B(1-p)
=\sum_{t=0}^{n-1}\Delta v^A(n-1-t)
  \theta(n-1-t\mid n-1,p)=\sum_{s=0}^{n-1}\Delta v^A(s)\theta(s\mid n-1,p)
=V_A(p).
\end{align*}
Thus the two players have the same prize spread at complementary battle-winning probabilities.

For $p\in(0,1)$, nonconstancy and monotonicity imply $V_i(p)>0$. Define the nonnegative weights
\[
\omega_{it}(p)=\frac{\Delta v^i(t)\theta(t\mid n-1,p)}{V_i(p)},
\qquad \sum_{t=0}^{n-1}\omega_{it}(p)=1.
\]
The binomial-derivative calculation in \autoref{OA-lem:V'V} gives
\[
p(1-p)\frac{V_i'(p)}{V_i(p)}
=\sum_{t=0}^{n-1}t\omega_{it}(p)-(n-1)p.
\]
Since this weighted mean of $t$ is at most $n-1$, the analogue of the bound in \autoref{OA-lem:V'V} is
\[
L_i(p)\equiv1-2p+p(1-p)\frac{V_i'(p)}{V_i(p)}
=1+\sum_{t=0}^{n-1}t\omega_{it}(p)-(n+1)p
\leq n-(n+1)p<n.
\]
Unlike identity independence, budget balance alone does not require symmetric increments within either player's schedule, so we use the bound $n$ instead of $\lceil n/2\rceil$.

\paragraph{Sufficiency.}
Constant own-prize schedules uniquely induce zero effort. For a nonconstant rule, the normalization in \autoref{subsec:unique} preserves equilibria up to effort rescaling, so suppose $v^i(0)=0$ and $v^i(n)=1$ for both players. The proposed efforts $x_i^*=r\bar p_A\bar p_B\mathcal V_{\mathrm{ID}}/c_i$ are positive. Their ratio gives winning probabilities $\bar p_i$, and the preceding identity gives $V_i(\bar p_i)=\mathcal V_{\mathrm{ID}}$.

As in the proof of \autoref{Theo:existence} (\OAref{app:existence-proof}), fix \emph{positive uniform opponent effort} $x_A^*>0$ and consider player $B$. Write $u=\ln x_B$ and
\[
p(u)=\frac{e^{ru}}{e^{ru}+(x_A^*)^r},
\qquad \frac{\partial p}{\partial u}=rp(1-p).
\]
The normalized payoff and its derivative are
\begin{align*}
\pi_B(u\mid x_A^*)&=\sum_{t=0}^{n-1}\Delta v^B(t)H(t+1\mid n,p(u))-nc_Be^u,\\
\pi_B'(u\mid x_A^*)&=nrp(1-p)V_B(p)-nc_Be^u.
\end{align*}
Because $V_B(p)>0$, the derivative has the same sign as
\[
G_B(u\mid x_A^*)
=\ln\!\big(rp(u)(1-p(u))V_B(p(u))\big)-u-\ln c_B.
\]
The bound just established yields
\[
G_B'(u\mid x_A^*)=rL_B(p(u))-1<rn-1\leq0,
\]
where the last inequality uses $r\leq1/n$.

By the candidate formula, $G_B$ is zero at $u=\ln x_B^*$. Since it is strictly decreasing, $\pi_B'$ is positive below this point and negative above it. Thus $x_B^*$ is the unique best response among positive uniform efforts. Strict increase in payoff over $0<x_B<x_B^*$, together with continuity at zero, excludes zero as another maximizer. The same argument applies to player $A$, so the candidate is an equilibrium of the restricted uniform-strategy game.

Both parts of the reduction in \autoref{lem:Uniform Strategy} extend to identity-dependent rules. Its proof in \OAref{app:uniform-proof} uses only monotonicity and nonconstancy of each player's own prize schedule for averaging, and budget balance and costs depending only on own strategy for interchangeability. Since $r\leq1/n\leq1$, part (i) of that reduction extends the candidate to the full game.

\paragraph{Uniqueness.}
We now follow the proof of \autoref{pro:unique} (\autoref{subsec:unique}). For a nonconstant rule, existence in the restricted game and part (ii) of the preceding reduction imply that every equilibrium of the full game is uniform, possibly mixed. Any pure uniform equilibrium must have both efforts positive: positive effort against zero can be reduced without changing the outcome, whereas two zero efforts permit a profitable sufficiently small positive deviation. The first-order conditions and the common prize spread give $c_Ax_A=c_Bx_B$, fixing the winning probabilities at $\bar p_A,\bar p_B$ and then fixing efforts at $x_A^*,x_B^*$. Hence there is no other pure equilibrium.

Suppose, toward a contradiction, that a mixed uniform equilibrium exists. By interchangeability, a randomizing player's equilibrium strategy is also a best response to the opponent's positive candidate effort. But the candidate effort is the unique best response to that effort, including among zero and positive uniform efforts. The player therefore cannot randomize---a contradiction. This proves uniqueness, including among mixed strategies. The normalization mapping extends existence and uniqueness to every nonconstant feasible rule, completing part (i).

\noindent\textit{Part (ii).} Fix $r>1/n$ and consider the feasible schedules
\[
v^B(k)=\mathbf 1\{k=n\},\qquad
v^A(k)=\mathbf 1\{k\geq1\}.
\]
We must rule out every pure equilibrium of the full game, not just the uniform candidate. The following necessity argument holds for every $r>0$ and does not use the sufficient bound in part (i).

\begin{samepage}
In any pure equilibrium, the positivity argument in the proof of \autoref{lem:pure-uniformity} (\OAref{app:pure-uniformity-proof}) applies because both own-prize schedules are nonconstant and nondecreasing, so all battle efforts are positive. Let $D_j$ denote the common marginal prize from winning battle $j$. The battle-specific first-order conditions give
\[
c_Ax_{A(j)}
=rp_{A(j)}p_{B(j)}D_j
=c_Bx_{B(j)}.
\]
Consequently,
\[
\frac{x_{B(j)}}{x_{A(j)}}=\frac{c_A}{c_B},
\qquad
p_{A(j)}
=\frac{1}{1+(c_A/c_B)^r}
=\bar p_A,
\qquad
p_{B(j)}=\bar p_B.
\]
Conditional independence then implies $D_j=V_A(\bar p_A)=V_B(\bar p_B)=\mathcal V_{\mathrm{ID}}$ for every $j$. Substituting into the first-order conditions yields $x_i^*=r\bar p_A\bar p_B\mathcal V_{\mathrm{ID}}/c_i,\ i\in\{A,B\}.$
\par\end{samepage}

For the displayed schedules, write $q=\bar p_B=1/[1+(c_B/c_A)^r]$. The common marginal prize at this candidate is $q^{n-1}$, and player $B$'s candidate payoff is
\[
\pi_B^*=q^n-nr(1-q)q^n=q^n[1-nr(1-q)].
\]
For sufficiently large $c_B/c_A$, this payoff is negative because $nr>1$, whereas zero effort gives zero against the opponent's positive effort. Thus no pure equilibrium exists for these cost pairs. This establishes sharpness as a universal guarantee over feasible rules and cost asymmetries.

\section{Power-form convex costs}\label{OA-app:costs}

We consider costs with a common curvature $\gamma\geq1$ and original Tullock exponent $R>0$, and write $\hat x_{i(j)}\geq0$ for original effort.

\subsection{Separable power costs}\label{OA-subsec:convex}

Let player $i$ incur cost $c_i\hat x_{i(j)}^\gamma$ in battle $j$, with $0<c_A\leq c_B$. The change of variables $x_{i(j)}=\hat x_{i(j)}^\gamma$ is one-to-one and onto on nonnegative efforts, makes costs linear, and, when total effort in the battle is positive, transforms the winning probability into
\[
\frac{x_{i(j)}^{R/\gamma}}{x_{A(j)}^{R/\gamma}+x_{B(j)}^{R/\gamma}};
\]
each player still wins with probability $1/2$ when both efforts are zero. The map preserves winning probabilities and costs for every effort vector, and hence expected payoffs for every pure or mixed strategy profile. The full game is therefore strategically equivalent to the linear-cost model with effective discriminatory power $r=R/\gamma$: for every cost ratio, the equilibrium is unique and pure under every feasible identity-independent rule when $0<r\leq b_n$, and under every feasible identity-dependent rule when $0<r\leq1/n$. No separate restriction on $R$ is needed.

\subsection{Aggregate power costs}\label{OA-app:alternative}\label{OA-subsec:alternative}

Now let player $i$ incur cost $C_i\hat X_i^\gamma$, where $\hat X_i=\sum_{j\in\mathcal N}\hat x_{i(j)}$ and $0<C_A\leq C_B$. When $\gamma>1$, effort in one battle raises the marginal cost of effort in the others, as when a firm allocates engineers across parallel R\&D projects.

Suppose $0<R\leq1$. Constant prize schedules uniquely induce zero effort. Under a nonconstant rule, complete averaging preserves total effort $\hat X_i$, and hence cost, realization by realization; against a uniform opponent, it weakly raises the expected prize, strictly for a nonuniform vector against positive opponent effort. The arguments for \autoref{lem:Uniform Strategy} in \OAref{app:uniform-proof} use only these properties, budget balance, and costs that depend only on the player's own strategy. Hence every equilibrium of the restricted uniform-strategy game is an equilibrium of the full game, and if the restricted game has an equilibrium, every equilibrium of the full game is uniform. We impose $R\leq1$ so that the averaging argument applies to original efforts.

For uniform strategies, set
\[
x_i=\hat x_i^\gamma,\qquad r=\frac R\gamma,\qquad c_i=n^{\gamma-1}C_i.
\]
Then $C_i(n\hat x_i)^\gamma=nc_ix_i$, $\hat x_i^R/(\hat x_A^R+\hat x_B^R)=x_i^r/(x_A^r+x_B^r)$, and $c_B/c_A=C_B/C_A$, so the restricted game coincides with its linear-cost counterpart; unlike the separable case, this equivalence applies only to the restricted game. Combined with the preceding reduction, the linear-cost results imply that, for every cost ratio, the full game has a unique, pure equilibrium under every feasible identity-independent rule when $r\leq b_n$, and under every feasible identity-dependent rule when $r\leq1/n$. These guarantees maintain the restriction $0<R\leq1$ used above.

\subsection{Prize design and the cost of neutrality}\label{OA-subsec:cost-design}

Write $c_i$ for the effective linear-cost coefficient: the original coefficient in \autoref{OA-subsec:convex} and $n^{\gamma-1}C_i$ in \autoref{OA-app:alternative}. Converting equilibrium efforts back to original units, within the respective equilibrium domains a nonconstant identity-independent rule induces
\begin{gather*}
\hat x_{i(j)}=\left(\frac{r\bar p_A\bar p_B\mathcal V}{c_i}\right)^{1/\gamma},
\qquad
\bar p_A=\frac{1}{1+(c_A/c_B)^r},\\
\widehat{\mathbf{TE}}=n(r\bar p_A\bar p_B)^{1/\gamma}
\big(c_A^{-1/\gamma}+c_B^{-1/\gamma}\big)\mathcal V^{1/\gamma};
\end{gather*}
for identity-dependent rules, replace $\mathcal V$ by $\mathcal V_{\mathrm{ID}}$. The prize spreads here use the actual schedules. Normalizing a nonconstant schedule with range $s\in(0,1]$ multiplies original efforts by $s^{-1/\gamma}\geq1$, so normalized rules also suffice for maximizing original total effort. The coefficient of $\mathcal V^{1/\gamma}$ does not depend on the rule, so maximizing original total effort is equivalent to maximizing $\mathcal V$ or $\mathcal V_{\mathrm{ID}}$, and the optimal rules in \autoref{thm:tiemargin} and \autoref{pro:headstart} carry over at the transformed winning probability. On the common domain $r\leq1/n$, with $R\leq1$ also imposed for aggregate costs, the coefficient cancels from the ratio of optimal original efforts, which equals $\rho(\bar p_A)^{1/\gamma}$. Holding $n,R,\gamma$ fixed, the proportional loss $1-\rho(\bar p_A)^{1/\gamma}$ therefore weakly increases with cost asymmetry and is zero exactly when $\rho(\bar p_A)=1$; for $n\geq2$, it tends to $1-2^{-1/\gamma}$ as $c_B/c_A\to\infty$.

Neither extension covers player-specific curvature, which generally breaks the prize-independent effort-ratio argument. With different curvatures $\gamma_i$, the transformation would give the players different contest exponents $R/\gamma_i$, rather than the common exponent needed for our linear-cost results.

\end{document}